\pdfoutput=1

\documentclass[11pt]{article}
\usepackage[utf8]{inputenc}
\usepackage[T1]{fontenc}
\usepackage[english]{babel}

\usepackage[totalwidth=480pt, totalheight=680pt]{geometry}

\usepackage{xcolor}
\colorlet{mygreen}{green!50!black}
\colorlet{myblue}{green!20!blue}
\colorlet{myred}{red}
\colorlet{myorange}{orange!80!red}

\usepackage{cite}
\usepackage{mathptmx}
\usepackage{amsmath,amsfonts,amsthm,amsbsy,amssymb}
\usepackage[hidelinks,%
 						colorlinks=true,%
 						linkcolor=myblue,%
 						citecolor=mygreen,%
 						urlcolor=myblue,%
 						bookmarksnumbered=true]{hyperref}
\usepackage{mathtools}
\usepackage{nicefrac}
\usepackage{mathrsfs}
\usepackage{braket}
\usepackage{bbm}
\numberwithin{equation}{section}
\usepackage{lmodern}
\usepackage[bottom]{footmisc}

\usepackage{tikz}
\usetikzlibrary{shapes}
\usepackage{floatrow}
\usepackage[labelfont=bf, font={small}]{caption}

\def\beq{\begin{equation}}
\def\eeq{\end{equation}}
\def\bea{\begin{eqnarray}}
\def\eea{\end{eqnarray}}
\def\bit{\begin{itemize}}
\def\eit{\end{itemize}}

\def\vp{\varphi}

\def\vp{\varphi}

\def\rI{{{\rm I}}}
\def\rL{{{\rm L}}}

\def\cJ{{\cal J}}
\def\cI{{\cal I}}

\def\cV{{\cal V}}

\def\md{{\mathfrak{d}}}

\def\mT{{\mathfrak{T}}}
\def\mJ{{\mathfrak{J}}}
\def\msl{{\mathfrak{sl}}}
\def\mL{{\mathfrak{L}}}
\def\mF{{\mathfrak{F}}}
\def\mH{{\mathfrak{H}}}
\def\mP{{\mathfrak{P}}}

\def\mTh{\hat{\mathfrak{T}}}

\def\ba{{\bf{a}}}

\def\bm{{\bf{m}}}

\def\br{{\bf{r}}}

\def\brr{{\bf{r}}_{-1}}

\def\bL{{\mathbf{\Lambda}}}

\def\bd{{\boldsymbol{\delta}}}

\def\bE{{\mathbf{1}}}
\def\bD{{\mathbf{3}}}
\def\bF{{\mathbf{5}}}

\def\Ae{{^{[1]\!}A}}
\def\Al{{^{[\ell]\!}A}}
\def\Bl{{^{[\ell]\!}B}}
\def\Az{{^{[2]\!\!}A}}
\def\Bz{{^{[2]\!}B}}
\def\Ad{{^{[3]\!}A}}
\def\Bd{{^{[3]\!}B}}
\def\Av{{^{[4]\!}A}}
\def\Bv{{^{[4]\!}B}}

\def\Aff{{A_1^{(1)}}}
\def\Vir{{\rm Vir}}
\def\Ch{{\rm Ch}\,}
\newcommand{\cLt}[2]{{}^{[#1]}\!\mathfrak{L}_{#2}^{\text{tower}}}
\newcommand{\cL}[2]{{}^{[#1]}\!\mathfrak{L}_{#2}^{\text{coset}}}
\newcommand{\sL}[2]{{}^{[#1]}\!\mathcal{L}_{#2}^{\text{sug}}}

\def\E10{{{\rm E}_{10}}}

\def\DD{{ {}_*^*}}

\newtheorem{theorem}{Theorem}

\newtheorem{corollary}{Corollary}

\newcommand{\eprint}[1]{{\href{http://arxiv.org/abs/#1}{\texttt{arXiv:#1}}}}
\newcommand{\eprintm}[1]{{\href{http://arxiv.org/abs/#1}{\texttt{arXiv:#1}}}}
\newcommand{\eprintN}[1]{{\href{http://arxiv.org/abs/#1}{\texttt{arXiv:#1 [hep-th]}}}}
\newcommand{\eprintRT}[1]{{\href{http://arxiv.org/abs/#1}{\texttt{arXiv:#1 [math.RT]}}}}
\newcommand{\doi}[2]{\href{http://dx.doi.org/#2}{#1}}

\begin{document}
\begin{titlepage}
\setcounter{page}{0}
\begin{center}
\vspace*{1.5cm}
\textbf{\LARGE From Tensor Algebras to \\[3mm]
Hyperbolic Kac-Moody Algebras}\\[2ex]
\vspace{2cm}
\textsc{\Large Axel Kleinschmidt$^1$, Hannes Malcha$^{1,2}$}\\[1ex]
\textsc{\Large and Hermann Nicolai$^1$}\\
\vspace{1cm}
${}^1$\textit{Max Planck Institute for Gravitational Physics (Albert Einstein Institute),}\\
\textit{14476 Potsdam, Germany}\\[2ex]
${}^2$\textit{Institute for Theoretical Physics, ETH Z\"urich, 8093 Z\"urich, Switzerland}\\
\vspace{2cm}
\textbf{Abstract}
\end{center}
\small{
\noindent
We propose a novel approach to study hyperbolic Kac-Moody algebras,
and more specifically, the Feingold-Frenkel algebra $\mF$, which is based
on considering the tensor algebra of level-one states before descending to
the Lie algebra by converting tensor products into multiple commutators. This
method enables us to exploit the presence of mutually commuting coset
Virasoro algebras, whose number grows without bound with increasing affine
level. We present the complete decomposition of the tensor algebra under the
affine and coset Virasoro symmetries for all levels $\ell\leq 5$, as well as the
maximal tensor ground states from which all elements of $\mF$ up to level five
can be (redundantly) generated by the joint action of the affine and coset 
Virasoro generators, and subsequent conversion to multi-commutators,
which are then expressed in terms of transversal and longitudinal DDF states.
We outline novel directions for future work.
}
\vspace{\fill}
\end{titlepage}

\tableofcontents

\newpage
\section{Introduction}
\label{sec:1}

In this paper, we continue our study \cite{CKMN} of the Feingold-Frenkel algebra
$\mF$ \cite{FF}. This is a hyperbolic Kac-Moody algebra (KMA) based on 
the indefinite rank-three Cartan matrix 
\beq\label{eq:Aij}
(A_{ij}) \,\equiv \, \br_i\cdot \br_j\,=\, 
\begin{pmatrix} 2 &-1 & 0 \\
-1 & 2 & -2\\
0 & -2 & 2\\
\end{pmatrix} 
\hspace{30mm}
\begin{picture}(100,20)
\thicklines
\put(10,5){\line(1,0){40}}
\put(50,10){\line(1,0){40}}
\put(50,0){\line(1,0){40}}
\put(55,5){\line(1,1){10}}
\put(55,5){\line(1,-1){10}}
\put(85,5){\line(-1,1){10}}
\put(85,5){\line(-1,-1){10}}
\put(10,5){\circle*{10}}
\put(50,5){\circle*{10}}
\put(90,5){\circle*{10}}
\put(3,-12){$-1$}
\put(47,-12){$0$}
\put(88,-12){$1$}
\end{picture}
\,,
\eeq 
where $\br_i \in \{ \brr, \br_0,\br_1\}$ are the simple roots and we have also 
shown the Dynkin diagram (see \cite{Kac} for an introduction to the theory
of KMAs). The algebra $\mF$ is the smallest hyperbolic KMA which admits
an affine  null root, $\bd = \br_0 \,+\, \br_1$, hence an affine subalgebra 
$\Aff$, the untwisted affine extension of $\mathfrak{sl}(2)$. The presence of
such a null root is an absolutely essential ingredient in our construction (as
for  instance DDF operators could not even be defined without it). 
The basic conventions and notation as well as the basic definitions 
and results relevant to the study of $\mF$ are explained at
length in \cite{CKMN}. Here we summarize only some essential 
features, and refer to that paper as well as to the seminal work of 
\cite{FF} for further details.

A key property of $\mF$ (as for every KMA) is that it can 
be written as a graded direct sum~\cite{FF}
\beq\label{eq:F}
\mathfrak{F} \,=\, \bigoplus_{\ell \in \mathbb{Z}} \, \mathfrak{F}^{(\ell)} 
\,=\, \mathfrak{F}_- \oplus \mathfrak{F}^{(0)} \oplus \mathfrak{F}_+
\eeq
in an affine level expansion with respect to its affine subalgebra 
$\Aff \equiv \mF^{(0)}$, where $\ell\in\mathbb{Z}$ 
are the eigenvalues of the central charge operator of $\Aff$.
In terms of affine representations,
the level-one sector $\mathfrak{F}^{(1)}$ 
is the highest weight basic representation of $\Aff$, and $\mF^{(-1)}$ 
is its conjugate lowest weight representation.
The full algebra $\mF$ is then generated by multiply commuting $\mF^{(1)}$,
and $\mF^{(-1)}$, respectively, to obtain $\mF_+$ and $\mF_-$.
Equivalently, for $\ell\geq 2$ the level-$\ell$ sector of $\mF$ 
in \eqref{eq:F} can be recursively generated from preceding levels by
\beq\label{eq:F1Fell}
\mathfrak{F}^{(\ell)} \,=\, \left[ \mathfrak{F}^{(1)} \,,\,
\mathfrak{F}^{(\ell-1)} \right]
\eeq
(idem for $\ell \leq -2$ and all negative levels); a proof of
this statement can for instance be found in appendix A of \cite{CKMN}.
Accordingly, in \cite{CKMN} we have defined the maps 
$\mathcal{I}_\ell \colon \mathfrak{F}^{(1)} 
\otimes \mathfrak{F}^{(\ell-1)} \to \, \mathfrak{F}^{(\ell)}$ by
\beq\label{eq:Iell}
\mathcal{I}_\ell (u \otimes v) \coloneqq [u,v]
\eeq
for $\ell\geq 2$, with $u \in \mathfrak{F}^{(1)}$ and 
$v \in \mathfrak{F}^{(\ell-1)}$,
such that we have the vector space isomorphism
\beq\label{eq:FKerI}
\mF^{(\ell)} \,\cong \, \mF^{(1)} \otimes \mF^{(\ell -1)} 
\Big/ \, {\rm Ker}\; \cI_{\ell} \, .
\eeq
This iterative construction of $\mF$ thus relies on obtaining the level-$\ell$ 
sector from the preceding level-$(\ell -1)$ sector of the algebra, and requires
the determination of the kernel ${\rm Ker}\;\cI_\ell$ at each step, and
becomes more and more cumbersome with increasing level. The
crucial fact for our construction is now that at each such 
step the tensor product of affine representations is accompanied by 
a new coset Virasoro algebra \cite{GKO,GO,CKMN} whose presence can 
be exploited for the evaluation of the relevant product of affine 
representations. However, when converting the product 
$\mF^{(1)} \otimes \mF^{(\ell -1)}$ to $\mF^{(\ell)}$ by taking 
the quotient by ${\rm Ker}\;\cI_{\ell}$ the Virasoro module structure is lost: 
the level-$\ell$ sectors of $\mF$ are no longer Virasoro representations. 
This fact has so far prevented the full exploitation of these Virasoro
symmetries for a better understanding of hyperbolic KMAs.

In order to avoid this drawback and the loss of the Virasoro structure, 
we propose a different procedure in this paper, by keeping the tensor products 
of the level-one sectors as long as possible, and descending from the 
tensor products to the Lie algebra only in the very last step. 
This procedure has the advantage that we can fully exploit the simultaneous
presence of an increasing number of independent and commuting Virasoro
algebras. Our construction puts in evidence the rapidly increasing 
complexity of $\mF$, as the number of coset Virasoro representations 
grows without bound with the affine level and for $\mF$ involves 
{\em all} minimal ($c<1$) representations of the Virasoro algebra. 
While the action of each of these coset Virasoro algebras was exhibited in 
\cite{CKMN} at each step, proceeding from level-$(\ell -1)$ to level-$\ell$, the
main advance reported in this paper is thus {\em the simultaneous and mutually
commuting realization of all these coset Virasoro algebras} up to any given
level, see \eqref{eq:VirAff} below. This is a crucial step, because the action
of only a single coset Virasoro algebra per level is not enough to generate 
all states in $\mF$, as we already pointed out in section 6.2 of \cite{CKMN}.
Keeping the tensor products until the very last step introduces
a `multi-string flavor' into our construction, which is vaguely reminiscent
of an old proposal on the possible realization of KMAs
in string theory \cite{Witten}, even though the final step from
the tensor product will lead us back to the one-string Fock space
description, as we will explain in section~\ref{sec:4}.
The unbounded `pile-up' of independent Virasoro modules for
$\ell\rightarrow\infty$, which was already noticed in \cite{CKMN}, 
is one of the most intriguing features of our construction, and indicates
that in order to understand the global structure of $\mF$ we will eventually 
need to consider the action  of an $\infty$--fold product of Virasoro algebras.

To exhibit this structure
we first study the tensor algebra $\mT$ based on tensor products
of the basic representation which occurs at level one. For an explicit
and convenient description of the level-one states which constitute the
basic representation of $\Aff$ we make use of the DDF formalism \cite{DDF}
in the form developed in \cite{GN,GN1}. Only after analyzing the tensor products
we return to the Lie algebra in a second step by converting tensor 
products into multi-commutators. In this process numerous elements of the
tensor algebra are mapped to zero, whence the level-$\ell$ subspaces
 $\mF^{(\ell)}$ of the hyperbolic Lie algebra itself are no
longer representations of the relevant coset Virasoro algebras.

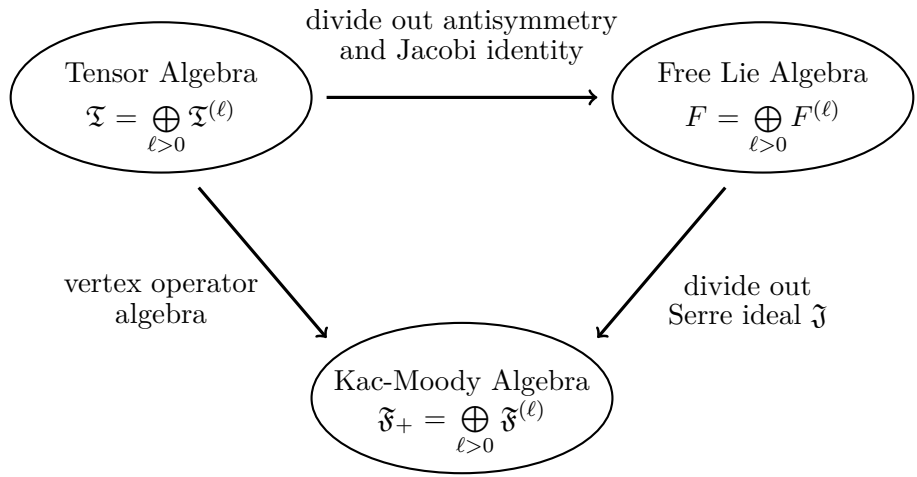
\begin{figure}[ht]
\begin{tikzpicture}
\draw [thick] (-4,0) ellipse (2 and 1);
\node at (-4,.3) {{Tensor Algebra}};
\node at (-4,-.4) {{$\mT \,=\, \bigoplus\limits_{\ell > 0} \mT^{(\ell)}$}};
\draw [thick] (4,0) ellipse (2 and 1);
\node at (4,.3) {{Free Lie Algebra}};
\node at (4,-.4) {{$F \,=\, \bigoplus\limits_{\ell > 0} F^{(\ell)}$}};
\draw [thick] (0,-4) ellipse (2 and 1);
\node at (0,-3.8) {{Kac-Moody Algebra}};
\node at (0,-4.4) {{$\mathfrak{F}_{+} \,=\, \bigoplus\limits_{\ell > 0} \, 
\mathfrak{F}^{(\ell)}$}};
\node at (0,1) {{divide out antisymmetry}};
\node at (0,.6) {{and Jacobi identity}};
\draw[->, very thick] (-1.8,0) to (1.8,0);
\node at (-4,-2.5) {{vertex operator}};
\node at (-4,-2.9) {{algebra}};
\draw[->, very thick] (-3.5,-1.2) to (-1.8,-3.2);
\node at (3.8,-2.5) {{divide out}};
\node at (3.8,-2.9) {{Serre ideal $\mJ$}};
\draw[->, very thick] (3.5,-1.2) to (1.8,-3.2);
\end{tikzpicture}
\caption{\label{fig:1}Different methods of obtaining (the positive half of) a 
hyperbolic Kac-Moody algebra.}
\end{figure}
It is instructive to compare our approach with the more 
traditional method of investigating hyperbolic KMAs 
\cite{FF,Kac,KMW,Kang1,Kang2,Kang3,BB}, which
we also review in section~\ref{sec:4.3}.
There one starts from the Free Lie Algebra, and then divides 
out level by level the relevant ideals associated to the Serre relations,
see fig. \ref{fig:1}. This approach becomes rather cumbersome already 
for very low levels ({\em to wit}, $\ell = 3,4,5$) and basically unmanageable 
for yet higher levels. By contrast, the present approach does not proceed 
in two steps, but goes directly from the tensor algebra of level-one 
elements to the KMA, rather than using the Free Lie Algebra as an 
intermediate construct. For this we make use of the vertex operator 
algebra (VOA) formalism \cite{Borcherds,IF,FLM} 
(see also \cite{GN,GN1}) to convert tensor products directly 
into multi-commutators. An important
advantage of the VOA formalism is that it takes care automatically of the 
Jacobi identities and the Serre relations. The main open problem is now
to combine the traditional study of $\mF$ with the DDF construction to find a
subset of the tensor algebra on which the map to $\mF$ is bijective. This would
lead to a full realization of the adjoint representation of $\mF$. 

We believe that our approach offers entirely new perspectives on the
study of hyperbolic KMAs, raising many new questions that can now be addressed:
\bit
\item
Some of the technical aspects of the present work that focusses on the 
hyperbolic Lie algebra $\mF$ are greatly facilitated by the fact 
that only minimal Virasoro models arise in the tensor algebra 
under consideration. This enables us to completely resolve the
tensor products into affine and minimal Virasoro representations for 
arbitrary levels, thanks to a key theorem of \cite{KW1}.
This is no longer the case for other hyperbolic Lie algebras, in particular  
for $\E10$ where only the three lowest levels obey $c_\ell < 1$.
While it is still true that the affine tensor products can be decomposed
into products of affine and Virasoro representations we lose
control over the $\mL_0$ eigenvalues for $c_\ell > 1$.
Hence $\E10$ as well as other higher rank hyperbolic KMAs
present qualitatively new complications beyond those 
of $\mF$ that will have to be addressed in future work.
\item As already pointed out in \cite{CKMN} the conversion of tensor
 products to Lie algebra elements leads to the `disappearance' of
 states, and hence to the loss of the Virasoro structure. As the 
 affine representations are not affected by this step, the problem
 of understanding $\mF$ now becomes one of elucidating the 
 `entanglement' of multiple Virasoro modules. Furthermore, the 
 conversion of tensor products to multi-commutators leads to the
 appearance of {\em longitudinal} DDF states, in accord with the fact
 that the string realization via the VOA construction is associated
 with a {\em sub-critical} bosonic string, as already pointed out
 in \cite{GN}. Understanding the full
 systematics of longitudinal states remains a problem for future study.
\item It has been known for a long time that affine representations of a given 
 level transform
 in a representation of the modular group \cite{Kac,KacPeterson}. 
 For $A_1^{(1)}$ there are 
 two representations at level one (related by an outer automorphism) that are 
 related by the modular group. However, only one of them appears in $\mF$. For 
 E$_8^{(1)}\subset {\rm E}_{10}$ there is only a single level-one representation 
 that is a singlet under the modular group. The extension of the modular group 
 action to the tensor algebras of the present paper and their fate when 
 descending to the hyperbolic KMAs is an interesting future avenue, as it
 will involve the simultaneous consideration of affine and multiple Virasoro
 modules. The connection to modular properties and Siegel modular 
 forms was a key motivation in the original work~\cite{FF}. The 
 decompositions \eqref{eq:T13} - \eqref{eq:T5} and their higher
 level generalizations suggest the existence of new Rogers-Ramanujan-like 
 identities that would follow from the results 
 in section~\ref{sec:3.1} and  their generalization to all
 levels.\footnote{We thank
 A.J. Feingold for pointing this out.}
\item Much of the previous mathematical literature has focused on the problem
 of determining root multiplicities, and finding appropriate special functions 
 with automorphic properties. Obtaining closed form expressions 
 for root multiplicities remains an outstanding problem, as there is so far
 not a single hyperbolic KMA for which the root multiplicities are known in
 closed form. Here, our approach to this problem offers a very different
 perspective, and furthermore allows for a much more explicit
 representation of the root space elements in terms of 
 (transversal and longitudinal) DDF states. The explicit
 expressions displayed in this paper and the presence of 
 `holes' appearing in the multiple Virasoro modules after 
 conversion of tensor products to multi-commutators suggest
 that the multiplicity formulas for levels $\ell\geq 3$ must have a more 
 complicated structure than the ones obtained so far (as is also apparent 
 from the results of \cite{Kang3,BB}). 
 \eit
 
 While these mathematical issues are of great importance in their own right,
 our main motivation continues to be the quest for a better understanding of
 the possible significance of $\mF$ and higher rank hyperbolic KMAs for 
 fundamental physics. Here we see at least three possible avenues for 
 further research and exploration. 
 
 First of all, the physical relevance of the affine subalgebra for General 
 Relativity has been known for a long time: $\Aff$ is known to be realized 
 as a generating symmetry (Geroch group) for axisymmetric
 stationary solutions of Einstein's equations in four dimensions (see 
 \cite{Julia,BM,Schladming} and references therein). This raises the question as 
 to the possible physical interpretation of the higher level sectors. Indeed the 
 search for {\em gauged} supergravities in two dimensions (a program that still
 needs to be completed) has revealed that the
 vector fields needed for the gauging belong to the basic representation
 of the relevant affine symmetry~\cite{Samtleben:2007an}, see 
also~\cite{Konig:2024psx} for a discussion in relation with supersymmetry.
Yet higher level representations might appear in connection with the tensor
hierarchy (see \cite{dWNS} for a review).
Related hints come from exceptional field 
theory~\cite{Bossard:2021jix,Bossard:2024gry}. 

Yet another perspective is furnished by a tensor hierarchy algebra construction 
that is 
likely to be relevant in a gauged supergravity or extended geometry context
with specific extensions to higher level affine representations that could occur 
in a supergravity context ~\cite{Cederwall:2021ymp}. 
The tensor hierarchy algebra constructed in~\cite{Cederwall:2021ymp} is 
different from $\mF$, but can also be partly obtained by quotienting a tensor 
algebra construction, such that the considerations of this paper could be 
relevant in this context as well.
The presence of an unbounded
number of Virasoro representations in the present approach
 may open new options for defining tensor hierarchies whose $p$-form 
 degrees of freedom are not limited by a finite number of space-time
dimensions, unlike for dimensionally reduced supergravities.

Thirdly, there exists a concrete proposal for M theory \cite{DHN0}, according to 
which space-time and concomitant field theoretic concepts emerge from 
a more fundamental `space-time-less' theory based on (the supersymmetric 
extension of) an $\E10/K(\E10)$ sigma model, where the space-time 
based dynamics is mapped to a null geodesic motion of a (spinning) point
particle 
on the $\E10/K(\E10)$ coset manifold. This proposal grew out of earlier
BKL-type investigations of the behavior of solutions of Einstein's equations
near a space-like (cosmological) singularity \cite{BKL}, and provides a more
systematic 
explanation for the appearance of chaotic metric oscillations {\em \`a la} BKL 
near the singularity. Namely, it can be shown that the metric oscillations can
be described in terms of a billiard that takes place in the fundamental
Weyl chamber of the hyperbolic KMA~\cite{Damour:2002et}. Our results
may also have implications for the quantum mechanical resolvability 
of the Big Bang singularity, see \cite{Nic} where it has been argued that 
the rapidly increasing complexity of expressions like those presented
in section~3.1 hints at an element of non-computability in
the approach to the singularity.

Let us finally remark that besides indefinite Kac-Moody algebras of hyperbolic
type
(like E$_{10}$ or $\mF$), Lorentzian algebras of {\em non}-hyperbolic type have
also been discussed 
as potential candidate symmetries of M theory or similar gravitational
theories~\cite{West:2001as,Cook:2013iya}. The rank-four Kac-Moody algebra $A_1^{+++}$ extending  the hyperbolic algebra $\mathcal{F}$ by one more node has been studied in particular in~\cite{Glennon:2020qpt,Boulanger:2022arw}.

\subsubsection*{Acknowledgements}
HM thanks the Max Planck Institute for Gravitational Physics in Potsdam as well
as the ETH Z\"urich for hospitality during various stages of this project. We
are grateful to A.J. Feingold and M.~Gaberdiel for discussions.

\section{Tensor algebra and coset Virasoro representations}
\label{sec:2}

The key idea of this paper is not to start with the KMA right away
but instead to analyze the tensor product of level-one states 
{\em before} converting the tensor products to multiple commutators.
The main reason for this is that the results involving coset Virasoro algebras
only hold for the tensor product, but are not directly valid for the Kac-Moody 
algebra whose level-$\ell$ sectors do not form representations of
the coset Virasoro algebra. 

For our DDF construction in particular, this has the advantage that we are
only working with multiple tensor products of level-one states, which are very
well understood in terms of transversal DDF states, and in particular do not
contain 
longitudinal DDF operators. The latter only appear when mapping the tensored 
DDF states from the tensor algebra to the Kac-Moody algebra.

\subsection{The tensor algebras \texorpdfstring{$\mT$}{T} and \texorpdfstring{$\mTh$}{That}}
\label{sec:2.1}

We define the tensor algebra 
$\mT = \bigoplus_{\ell \in \mathbb{N}} \mT^{(\ell)}$ generated by the 
level-one module $L \equiv L(\bL_0 + 2 \bd)$, {\em alias} the basic 
representation, which is a highest weight affine representation at 
level one in our conventions (see \eqref{eq:Lambda2} below for our 
nomenclature regarding weights). There is an analogous description
for the conjugate (negative level) representations.
The graded pieces of the tensor algebra are
\begin{align}\label{eq:mTn}
\begin{aligned}
&\mT^{(1)} \equiv \mF^{(1)} \,=\, L \, , \\
&\mT^{(2)} \,=\, L \otimes L \, , \\
&\mT^{(3)} \,=\, L \otimes L \otimes L \, , \\
&\mT^{(4)} \,=\, L \otimes L \otimes L \otimes L \, , \\
&\mT^{(5)} \,=\, \ldots
\end{aligned}
\end{align}
$\mT^{(\ell)}$ consists of level-$\ell$ affine representations by construction.
All our results in section~\ref{sec:2} will be valid for $\mT$, but for the
conversion to the KMA only the following subalgebra $\mTh\subset\mT$ will be 
relevant 
\begin{align}\label{eq:mTnh}
\begin{aligned}
&\mTh^{(1)} \equiv \mT^{(1)} \,=\, L \, , \\
&\mTh^{(2)} \,=\, L \wedge L \, , \\
&\mTh^{(3)} \,=\, L \otimes L \wedge L \, , \\
&\mTh^{(4)} \,=\, L \otimes L \otimes L \wedge L \, , \\
&\mTh^{(5)} \,=\, \ldots
\end{aligned}
\end{align}
because under the commutator map \eqref{eq:Iell} the symmetric product
$S^2(\mT^{(1)})\subset \mT^{(1)} \otimes \mT^{(1)}$ is mapped to zero and only the antisymmetric (wedge) product survives.\footnote{We use the notation 
$a \wedge b = a \otimes b - b \otimes a$.} Working with $\mTh$ rather than $\mT$
shortens some expressions.

For each level-$\ell$ affine representation there is an associated level-$\ell$
Sugawara Virasoro algebra
\cite{GO} whose action on an element $w\in\mT^{(\ell)}$ is defined by
\beq\label{eq:Sugawara}
\sL{\ell}{m} \,(w)\, 
\coloneqq \, \frac{1}{2(\ell + 2)} \sum_{n\in\mathbb{Z}} \DD T_{nA} T^A_{m-n} \DD \, (w) 
\,\equiv\, \frac{1}{2(\ell + 2)} \sum_{n\in\mathbb{Z}} \DD T_n T_{m-n} \DD \, (w)
\eeq
with the affine generators 
\beq\label{TEFH}
T_m^A  \;\in\; \big\{ E_m\,,\, F_m\,,\, H_m \big\} \; ,
\eeq
see \cite{CKMN} for more details (we often omit $\mathfrak{sl}(2)$ indices 
for simplicity); by $\DD \cdots\DD$ we designate the
usual normal ordering prescription. The level-$\ell$ Sugawara 
operators come with the central charge \cite{GO}
\beq
c_\ell^\mathrm{sug} \,=\, \frac{3\ell}{\ell+2} \, .
\eeq
To evaluate the action of $\sL{\ell}{m}$ on a tensor product 
$u_1\otimes\cdots\otimes u_\ell \in \mT^{(\ell)}$ we use
\begin{align}
\begin{aligned}\label{eq:Lsug}
\sL{\ell}{m} (u_1\otimes\cdots\otimes u_\ell) \,=\, \frac{1}{2(\ell + 2)} 
\bigg(&\sum_{i=1}^\ell \sum_{n\in\mathbb{Z}}
u_1 \otimes \cdots \otimes \DD T_nT_{m-n} \DD (u_i) 
\otimes \cdots \otimes u_{\ell} \\
&\,+\, \sum_{1\leq i \neq j \leq \ell} \ \sum_{n\in\mathbb{Z}} 
u_1\otimes \cdots \otimes T_n u_i \otimes\cdots\otimes T_{m-n} u_j \otimes
 \cdots \otimes u_{\ell} \bigg)
\end{aligned}
\end{align}
as follows directly from the distributivity of the affine generators 
$T_m$, {\em i.e.}
\beq
T_m (u \otimes v) \,=\, (T_m u) \otimes v \,+\, u \otimes (T_m v) \, .
\eeq
In particular, for the mixed terms we can drop the normal ordering symbols.

For each level we next introduce a coset Virasoro algebra whose 
action on the tensor product $u\otimes v$ with $u\in\mT^{(1)}$ and 
$v\equiv u_1\otimes\cdots\otimes u_{\ell -1} \in\mT^{(\ell -1)}$ is defined by 
\cite{GKO} 
\beq\label{eq:cosetVir}
\cL{\ell}{m} (u\otimes v) \,\coloneqq\, \sL{1}{m} u \otimes v 
\,+\, u \otimes \sL{\ell-1}{m} v
 \,-\, \sL{\ell}{m}(u\otimes v) \, .
\eeq
It is important that the definition \eqref{eq:cosetVir} does {\em not} extend to 
$\mF^{(\ell)}$ when tensor products are replaced by multiple commutators. The
main reason for this is that a commutator may vanish even though the tensor 
product does not, when $\cI_\ell(u\otimes v) = [u,v] =0$. This is the primary
reason why we first study the tensor product $\mT^{(\ell)}$ before proceeding to
$\mF^{(\ell)}$.

Moreover, \eqref{eq:cosetVir} allows for an immediate generalization to the 
action of all level $k \le \ell$ coset Virasoro algebras on a level-$\ell$ state 
$ u_1\otimes\cdots\otimes u_{\ell} \in\mT^{(\ell)}$ via
\beq\label{eq:cosetVir2}
\Big( \underbrace{\bE \otimes \cdots \otimes \bE}_{(\ell-k) \text{ times}} 
\otimes \cL{k}{m} \Big)
(u_1\otimes\cdots\otimes u_{\ell}) \,=\,
u_1 \otimes \cdots \otimes u_{\ell-k} \otimes \cL{k}{m}
\big(u_{\ell-k+1} \otimes (u_{\ell-k+2} \otimes \cdots \otimes u_{\ell})\big)
\, .
\eeq
The key property of the coset Virasoro operators \eqref{eq:cosetVir} and
\eqref{eq:cosetVir2} is that they commute with the affine action~\cite{GKO}
\beq\label{eq:LmT}
\Big[\,\cL{\ell}{m} \,,\, T_n^A \, \Big] \,=\, 0
\eeq
for each level (as follows directly from the distributivity of the 
affine action). The commutativity extends to `composite' operators
such as those appearing on the r.h.s. of \eqref{eq:Lsug}. The level-$\ell$ coset 
Virasoro algebra comes with the central charge of a minimal model \cite{GKO,CKMN}
\beq\label{eq:cell}
c_\ell^{\mathrm{coset}} \,=\, 1 - \frac{6}{(\ell+1)(\ell +2)} \, .
\eeq
Therefore the $\cL{\ell}{0}$-eigenvalues $h_{r,s}^{(\ell)}$ of the extremal
coset Virasoro states must belong to the set (see {\em e.g.}
\cite{DiFrancesco})\footnote{These coset Virasoro modules appearing in the
tensor product of affine highest weight representations count the
multiplicities of identical affine representation that are shifted from the
extremal one by \textit{negative} multiples of the null root $\bd$. As in our
convention the characters of Virasoro modules are $q$-series and $q=e^{-\bd}$,
this means that the extremal state in the Virasoro module is annihilated by
$\cL{\ell}{m}$ for $m>0$.}
\beq\label{eq:hrs}
h_{r,s}^{(\ell)} \,\in\, H^{(\ell)} \quad \mbox{with} \quad
H^{(\ell)} 
\coloneqq
\left\{ \frac{[(\ell+2)r-(\ell+1)s]^2-1}{4(\ell+1)(\ell+2)} \; \bigg|\;
r= 1 , \ldots , \ell; \ s = 1, \ldots, r \right\}\, .
\eeq
Eqn. \eqref{eq:cell} shows that {\em all} minimal representations 
will occur in the analysis of $\mF$.

Let us mention that instead of \eqref{eq:cosetVir} one could contemplate 
a more general definition of $\cL{\ell}{m}$
by grouping the tensor product $u_1\otimes\cdots\otimes u_\ell$ into two factors 
$ u \otimes v \equiv (u_1\otimes\cdots\otimes u_k) 
\otimes (u_{k+1} \otimes \cdots \otimes u_\ell)$ 
with $k>1$ and replacing the r.h.s. of \eqref{eq:cosetVir} by
$\sL{k}{m} u \otimes v \,+\, u \otimes \sL{\ell-k}{m} v \,-\, \sL{\ell}{m}(u\otimes v)$. 
The result will then depend on the choice of $k$ and imply central charge
values different from \eqref{eq:cell} (in particular also with values $c>1$, 
{\em e.g.} for $\ell=6$ and $k=2$).
However, in that case the analog of \eqref{eq:F1Fell} is no longer true as
$[\mF^{(k)} , \mF^{(\ell -k)}]$ in general is a {\em proper} subspace of 
$\mF^{(\ell)}$ for $k>1$, hence fails to capture all of $\mF^{(\ell)}$. We give 
an explicit example of this phenomenon in appendix \ref{app:example}, using DDF
states. For this reason we will only make use of `consecutive' coset Virasoro
algebras with $k=1$ as in \eqref{eq:cosetVir}.

The tensor product of the level-one and any level-$(\ell-1)$ affine module is
the direct sum of tensor products of minimal model Virasoro representations and
level-$\ell$ affine modules. An explicit formula for this tensor product was first
given in Theorem 4.1 of \cite{KW1}. Here we recall the result of \cite{KW1} in the 
notation of this paper and additionally introduce the appropriate $\bd$-shift on 
the right-hand side. Let
\beq\label{eq:Kmn}
K_{m,n} \,=\, \left\{ k \in \mathbb{Z} \ \bigg\vert \ 
- \frac{1}{2} (m+1) \le k \le n \right\}
\eeq
and for $k \in K_{m,n}$
\beq\label{eq:rs}
r_{m,n,k} = \begin{cases} 2n + 1 \quad &k \ge 0 \\ m+1 \quad &k < 0 \end{cases} \, , 
\hspace{1cm}
s_{m,n,k} = \begin{cases} 2n + 1 - 2k \quad &k \ge 0 \\ m+2 + 2k \quad &k < 0 \end{cases} \, .
\eeq
Then the product of the level-one affine module and any level-$(\ell-1)$ affine
module is given by
\begin{align}\label{eq:prodL}
\begin{aligned}
L(\bL_0 + 2\bd) &\otimes L(m \bL_0 + 2n \bL_1 + p \bd) \\[2mm]
&\,=\, \bigoplus_{k \in K_{m,n}} 
\Vir(c_{\ell}^{\mathrm{coset}}, h_{r,s}^{(\ell)}) 
\otimes L\Big((m+1+2k) \bL_0 + 2(n-k) \bL_1 + (p + 2 - k^2) \bd \Big) \, , 
\end{aligned}
\end{align}
where $\ell = m + 2n + 1$, $r \equiv r_{m,n,k}$, $s \equiv s_{m,n,k}$ and 
$h_{r,s}^{(\ell)}$ as in \eqref{eq:hrs}. See \cite{KW1} for a proof of this 
equation. The matching of the characters in \eqref{eq:prodL}, which fixes the
$\bd$-shifts, is shown explicitly in appendix \ref{app:delta}. It is important 
to point out that each of the $\lfloor \frac{\ell}{2} \rfloor +1$ different 
level-$\ell$ affine modules  appears exactly once on the right-hand side of 
\eqref{eq:prodL}. We recall that the (hyperbolic) fundamental weights of $\mF$ 
appearing in these expressions 
obey $\br_i\cdot\bL_j = \delta_{ij}$ in our conventions, and are given by
\beq\label{eq:Lambda2}
\bL_{-1} \,=\, - \bd\,\,, \quad
\bL_{0} \,=\, - \brr - 2\bd\,\,, \quad
\bL_{1} \,=\, - \brr - 2\bd + \frac12 \br_1 \, .
\eeq
With \eqref{eq:prodL} at hand it is straightforward to compute $\mT^{(\ell)}$ 
to any level. We give the first five levels of $\mTh^{(\ell)}$ explicitly in 
section \ref{sec:3.1}.

For $E_{10}$ and other rank > 3 hyperbolic KMAs
the central charges $c_\ell^\mathrm{coset}$ which appear 
on the r.h.s. of \eqref{eq:prodL} are not bounded from above by one; 
{\em e.g.} for $\E10$ we have $c_\ell^\mathrm{coset} > 1$ for $\ell\geq 4$. Thus
the Virasoro eigenvalues are no longer elements from the finite set 
\eqref{eq:hrs}, and there is no restriction anymore on the admissible 
values of $h^{(\ell)}$ (other than $h^{(\ell)} > 0$).
In this case we are not aware of an equation similar to 
\eqref{eq:hrs} for the associated coset Virasoro eigenvalues.

\subsection{Implementing simultaneous coset Virasoro actions}
\label{sec:2.2}

In \cite{CKMN} we considered only the action of one coset Virasoro algebra
at a time. We are now interested in simultaneous consecutive actions of 
different coset Virasoro algebras and the question how they are related.
This is the step that requires sticking with the tensor products \eqref{eq:mTn}
and \eqref{eq:mTnh}.
The following theorem is a central result as it establishes the 
commutativity of such consecutive coset Virasoro actions, and
formally expresses a fact which is already evident from 
the repeated iteration of taking products in \eqref{eq:prodL}:
\begin{theorem}\label{th:main}~\\
Let $u_1 , u_2 \in \mT^{(1)}$ and $v \in \mT^{(\ell -2)}$; then for any 
$\ell\geq 3$ and $m,n \in \mathbb{Z}$
\beq\label{eq:Theorem}
\cL{\ell}{m} \left(u_1 \otimes \cL{\ell-1}{n} 
(u_2 \otimes v) \right) \,=\, \left(\bE \otimes \cL{\ell-1}{n}\right) 
\left( \cL{\ell}{m} (u_1 \otimes ( u_2 \otimes v)) \right) \, .
\eeq
\end{theorem}
\begin{proof}~
Starting from the left-hand side of \eqref{eq:Theorem}, we deduce
\begin{align}
\begin{aligned}
\cL{\ell}{m} \left(u_1 \otimes \cL{\ell-1}{n} 
( u_2 \otimes v) \right) 
&\,=\, \sL{1}{m} u_1 \otimes \cL{\ell-1}{n} ( u_2 \otimes v) \\
&\quad \,+\, u_1 \otimes {\color{myred} \sL{\ell-1}{m} \ \cL{\ell-1}{n} } 
( u_2 \otimes v) \\
&\quad \,-\, \sL{\ell}{m} \left( u_1 \otimes \cL{\ell-1}{n} 
( u_2 \otimes v) \right) \, .
\end{aligned}
\end{align}
The {\color{myred}red} operators commute by \eqref{eq:LmT}. For the right-hand
side we have
\begin{align}
\begin{aligned}
&(\bE \otimes \cL{\ell-1}{n}) \left( \cL{\ell}{m} (u_1 \otimes 
(u_2 \otimes v)) \right) \\
&\,=\,(\bE \otimes \cL{\ell-1}{n}) \left( \sL{1}{m} u_1 \otimes 
( u_2 \otimes v) \right) \\
&\quad \,+\, (\bE \otimes \cL{\ell-1}{n}) \left( u_1 \otimes \sL{\ell-1}{m} 
( u_2 \otimes v \right) \\
&\quad \,-\, (\bE \otimes \cL{\ell-1}{n}) \sL{\ell}{m} 
\left( u_1 \otimes ( u_2 \otimes v) \right) \\
&\,=\, \left( \sL{1}{m} u_1 \otimes \cL{\ell-1}{n} 
( u_2 \otimes v) \right) \\
&\quad \,+\, \left(u_1 \otimes \cL{\ell-1}{n} \sL{\ell-1}{m} 
( u_2 \otimes v) \right) \\
&\quad \,-\, (\bE \otimes \cL{\ell-1}{n}) \sL{\ell}{m} 
\left( u_1 \otimes ( u_2 \otimes v) \right) \, .
\end{aligned}
\end{align}
Then only the last term needs to be checked
\begin{align}
\begin{aligned}
&(\bE \otimes \cL{\ell-1}{n}) \sL{\ell}{m} 
\left( u_1 \otimes ( u_2 \otimes v) \right) \\
&\,=\, \frac{1}{2(\ell+2)} (\bE \otimes \cL{\ell-1}{n}) 
\Bigg( \sum_k \DD T_k T_{m-k} \DD u_1 
\otimes ( u_2 \otimes v)) \\
&\quad \,+\, \sum_k u_1 \otimes \DD T_k T_{m-k} \DD 
( u_2 \otimes v) 
\,+\, 2 \sum_k T_k u_1 \otimes T_{m-k} ( u_2 \otimes v)) \Bigg) \\
&\,=\, \frac{1}{2(\ell+2)} \Bigg( \sum_k \DD T_k T_{m-k} \DD 
u_1 \otimes \cL{\ell-1}{n} ( u_2 \otimes v)) \\
&\quad \,+\, \sum_k u_1 \otimes \cL{\ell-1}{n} \DD T_k T_{m-k} \DD
( u_2 \otimes v) \\
&\quad \,+\, 2 \sum_k T_k u_1 \otimes \cL{\ell-1}{n} T_{m-k} 
( u_2 \otimes v) \Bigg) \\
& \,=\, \sL{\ell}{m} \left( u_1 \otimes \cL{\ell-1}{n} 
( u_2 \otimes v) \right) \, .
\end{aligned}
\end{align}
\end{proof}
We stress that the commutativity relation \eqref{eq:LmT} is crucial for 
this proof. The theorem then implies the following corollary, whose 
proof is entirely analogous to the proof just given.
\begin{corollary}~\\
For $u_1, \ldots , u_{k+1} \in \mathfrak{T}^{(1)}$, $v\in \mT^{(\ell -k -1)}$,
$\ell \ge 3$, 
and for $k\leq \ell -2$ and $m,n \in \mathbb{Z}$
\begin{align}
\begin{aligned}
\cL{\ell}{m} &\left(u_1 \otimes \left(u_2 \otimes \ldots \otimes 
u_k \otimes \cL{\ell-k}{n} ( u_{k+1} \otimes v)\right)\right) \\
&\quad \,=\, (\bE \otimes \ldots \otimes \bE \otimes \cL{\ell-k}{n}) 
\left( \cL{\ell}{m} ( u_1 \otimes (u_2 \ldots 
\otimes u_{k+1} \otimes v ))\right) \, .
\end{aligned}
\end{align}
\end{corollary}
We repeat that throughout these computations we always employ
the `consecutive' definition \eqref{eq:cosetVir} corresponding
to the split $k = 1 + (k-1)$ for $k\leq \ell$. We have thus shown 
that $[ \cL{k}{m} , \cL{k'}{n}] = 0$ for all $k\neq k' \in \{2,\dots,\ell\}$,
that is, different coset Virasoro algebras commute.

The fact that the actions of the coset Virasoro algebras also commute
for distinct pairs of such algebras enables us to implement
a simultaneous and commuting action of $(\ell -1)$ Virasoro
algebras and the affine subalgebra $\Aff$ on $\mT^{(\ell)}$: 
\beq\label{eq:VirAff}
\underbrace{{\rm Vir} \,\oplus \cdots \,\oplus {\rm{Vir}}}_{(\ell -1)\,\, \text{times}}
\,\oplus\,\Aff \; : \quad \mT^{(\ell)} \rightarrow \; \mT^{(\ell)}\,,
\eeq
where all summands commute. In fact, this is already evident from
\eqref{eq:prodL} when one iterates such multiplication of affine
representations. This result puts in evidence the
unbounded `pile-up' of Virasoro algebras and representations 
with increasing level that was already pointed out in \cite{CKMN}, and
indicates that in order to understand the global structure of $\mF$
we will eventually need to consider the action of an $\infty$--fold 
product of Virasoro algebras on $\mTh$.

As we will see in section~\ref{sec:4}, this nice 
structure \eqref{eq:VirAff} is lost when 
descending from the tensor algebra to the KMA because of the appearance 
of `holes', whence an action of these Virasoro algebras cannot be implemented
on the KMA. From the present perspective, this is the main complication in 
understanding the structure of $\mF$.

\subsection{Realizing \texorpdfstring{$\mT^{(\ell)}$}{Tell} in terms of transversal DDF states}
\label{sec:2.3}

For the concrete calculation we make use of the Frenkel-Kac vertex operator 
construction \cite{FK}, but in the more convenient form based 
on transversal DDF states introduced in \cite{GN}, and used 
extensively in \cite{CKMN}. The transversal DDF operators \cite{DDF}
act on a (small) subspace of a Fock space $\mH$ of physical states 
associated to a fully compactified subcritical bosonic string in three 
Lorentzian target space dimensions, which we will define properly below in 
\eqref{eq:mH}.
The basic representation $L\equiv\mT^{(1)}$, {\em alias} the set of all 
level-one states, is then the linear span of all transversal DDF states:
\beq\label{eq:DDF1}
\left\{ \Ae_{-m_1} \ldots \Ae_{-m_n} \ket{\ba_w^{(1)}} \right\} 
\quad \text{with} \quad \md \,=\, \sum_{i=1}^n m_i \,+\, w^2 \, ,
\eeq
where $w$ is the weight, and the depth $\md$ counts the coefficient 
of $(-\bd)$ in any such DDF state. The explicit expressions of the transversal 
(and longitudinal) DDF operators, which depend on 
the level $\ell$, are given in section~3 of \cite{CKMN}.
The level-one tachyonic momenta in the tachyonic ground states
$|\ba_w^{(1)}\rangle$ are
\beq\label{eq:tachyon1}
\ba_w^{(1)} \,=\, - \brr - w^2 \bd \,+\, w \br_1 \, , \qquad w\in\mathbb{Z}
\eeq 
and are in one-to-one correspondence with the maximal weights w.r.t.
the Heisenberg subalgebra of $\Aff$ appearing in the weight diagram 
of the basic representation. The full set of weights in the basic representation 
associated to such a DDF state therefore consists of the level-one roots
\beq
\br \,=\, - \brr \,-\, \left( w^2 \,+\, \sum_{i=1}^n m_i\right) \bd \,+\, w \br_1
\eeq
or, equivalently, in the bracket notation of \cite{CKMN},
\beq
\br\,=\, - \left[ 1\,,\, w^2 \,+\, \sum_{i=1}^n m_i \,,\, w^2 \,+\, \sum_{i=1}^n m_i \,-\,w \right] \, .
\eeq
The corresponding root multiplicities are then simply given
by the transversal partition function, that is, mult$(\br) = 
p(1 -\br^2/2) \equiv \vp^{-1}(1 - \br^2/2)$ \cite{FF,Kac}, where $\varphi(q)=\prod_{k>0} (1-q^k)$ is the Euler function. 
That these states indeed form a representation
of the affine algebra $\Aff$ is guaranteed in our approach by
the explicit expressibility of the affine
generators in terms of transversal DDF operators \cite{GN1,CKMN}.
From the above representation it is also clear that any level-one 
state of weight $w$ and depth $\md$ must satisfy $w^2 \le \md$, 
which restricts the $\mathfrak{sl}(2)$ representations that can appear 
at a given depth. It is then straightforward to list the low depth elements
of $L$ with their multiplicities, {\em cf.} eqn. (3.28) in \cite{CKMN}.

Given this explicit description of $\mT^{(1)}$, any element in $\mT^{(\ell)}$ 
can be represented 
by a sum of $\ell$-fold tensor products of such level-one DDF states, {\em viz.} 
\beq\label{eq:Fock1}
\psi^{\otimes (\ell)} \, =\,
\sum_\nu a_\nu \, u_1^{(\nu)} \otimes \cdots \otimes u_\ell^{(\nu)} \; \in \; \mT^{(\ell)} \, ,
\eeq
where each $u_i^{(\nu)}$ is a level-one state of the form \eqref{eq:DDF1}.
We thus denote by $\psi^{\otimes(\ell)}$ an element of the tensor algebra 
subspace $\mT^{(\ell)}$, with the symbol $\otimes$ as a mnemonic device, to
distinguish it from the Lie algebra element $\psi^{(\ell)}$ obtained after
converting tensor products into multiple commutators (we will 
reserve capital letters $\Psi^{\otimes(\ell)}\in \mT^{(\ell)}$ 
for maximal tensor ground states, see below).
While the elements of $\mT^{(1)}$ 
are thus associated with specific one-string states, $\mT$ in this way 
becomes associated with a {\em multi-string Fock space}
of transversal DDF states built on the tachyonic momenta
\eqref{eq:tachyon1}. This realization is somewhat reminiscent of an
interpretation proposed in \cite{Witten}, but the analogy should 
not be taken too literally as we cannot assign any (bosonic 
or fermionic) statistics to such multi-string states, where in fact
almost all mixed symmetry types will appear.\footnote{However,
 Young tableaux techniques appear to be only of limited use for studying
 products of affine or Virasoro representations.} In the
following section we will explain how the map from $\mT^{(\ell)}$
to $\mF^{(\ell)}$ will take us back to a one-string Fock space.

\subsection{Maximal tensor ground states}
\label{sec:2.4}

In order to analyze the decomposition of $\mT^{(\ell)}$ into 
a sum of representations of \eqref{eq:VirAff} we introduce the notion
of a {\bf maximal tensor ground state (MTG)}, which by definition is an 
affine ground state and a simultaneous ground state for all the relevant 
coset Virasoro algebras. The important fact here is that at each level there are
only {\em finitely many} such MTGs which completely
characterize $\mT^{(\ell)}$, in the sense that any element of $\mT^{(\ell)}$
can be reached by the combined action of the affine and Virasoro operators
on the MTGs. Since all elements of $\mF^{(\ell)}$ follow from 
conversion of tensor products to multiple commutators, we are thus
able to reach all the elements of $\mF^{(\ell)}$, though in a 
highly redundant manner. Elements of the 
hyperbolic KMA that are obtained by following the `vertex operator 
algebra' arrow in figure~\ref{fig:1} will thus be denoted by $\psi$. 
The map from $\mT^{(\ell)}$ to $\mF^{(\ell)}$
will be discussed in detail in section~\ref{sec:4}. 

For a non-vanishing element $\Psi^{\otimes(\ell)}\in \mT^{(\ell)}$ to be a MTG,
it has to satisfy the following three conditions:
\begin{enumerate}
\item $\Psi^{\otimes (\ell)}$ is an {\bf affine ground state}, that is, 
\beq\label{eq:cond1}
E_m \, \Psi^{\otimes(\ell)} \,=\, F_m \, \Psi^{\otimes(\ell)} \,=\, 
H_m \, \Psi^{\otimes(\ell)} \,=\, 0 \quad \mbox{for all $m\geq 1$}
\quad\text{and}\quad
E_0 \, \Psi^{\otimes(\ell)} \,=\,0\, .
\eeq 
The second condition means that $\Psi^{\otimes (\ell)}$ is an $\mathfrak{sl}(2)$ highest weight state. Sometimes we will refer to the whole finite-dimensional $\msl(2)$ multiplet generated by $\Psi^{\otimes (\ell)}$ as a `ground state
multiplet'.  The four conditions 
in~\eqref{eq:cond1} are equivalent to $E_0 \Psi^{\otimes(\ell)} =F_1 \Psi^{\otimes(\ell)} =0$ in the notation of~\cite{CKMN}.
\item $\Psi^{\otimes (\ell)}$ is a \textbf{full Virasoro ground state} w.r.t. 
all coset Virasoro algebras of lower level, {\em i.e.}
\beq\label{eq:cond2}
\big(\underbrace{\bE \otimes \ldots \otimes \bE}_{(\ell-k) \text{ times}} 
\otimes \ \cL{k}{m}\big)
\Psi^{\otimes (\ell)} \,=\, 0
\eeq
for all $2 \le k \le \ell $ and $m > 0$.
\item $\Psi^{\otimes (\ell)}$ is a \textbf{full Virasoro eigenstate} w.r.t.
to all coset Virasoro algebras of lower level, {\em i.e.} there
exists a set of $\ell -1$ eigenvalues $h^{(k)} \in H^{(k)}$ such 
that
\beq\label{eq:cond3}
\big(\underbrace{\bE \otimes \ldots \otimes \bE}_{(\ell-k) \text{ times}} 
\otimes \ \cL{k}{0} \big)
\Psi^{\otimes (\ell)} \,=\, h^{(k)} \Psi^{\otimes(\ell)}
\eeq
for $2 \le k \le \ell$.
\end{enumerate}
The set of eigenvalues $\{h^{(2)}, \ldots, h^{(\ell) } \}$
is read off from the decomposition of $\mT^{(\ell)}$ into affine modules and
Virasoro algebras, {\em cf.} \eqref{eq:prodL} and \eqref{eq:T13}--\eqref{eq:T5}. Similarly, the
$\mathfrak{sl}(2)$ representation, depth and weight of each MTG are also read
off from the affine modules in the decomposition $\mT^{(\ell)}$. The number of
such MTGs is finite at each level $\ell$, and can be determined from a 
simple induction argument, using the fact that the multiplicities on the
r.h.s. of \eqref{eq:prodL} are all equal to one. Namely, let $n_\ell$, with 
$n_1 = n_2 = 1$, be the number of MTGs on level $\ell$. Then the 
number of MTGs on level $\ell+1 > 2$ is
\beq
n_{\ell+1} \,=\, \left( \left\lfloor \frac{\ell+1}{2} \right\rfloor \,+\,1 \right)  n_\ell \, .
\eeq
Let us repeat that it is the
commutativity of $\cL{k}{m}$ and $\cL{k^\prime}{m}$ for $k \not= k^\prime$ 
that enables us to formulate the second and third conditions in the above list
for the action of one $\cL{k}{m}$ at a time. Without Theorem \ref{th:main} and 
its Corollary this structure would be completely obscured.

For the labeling of MTGs we employ the general notation
introduced in \eqref{eq:Fock1}, by
appending three extra labels to $\Psi^{\otimes(\ell)}$, namely
\beq\label{eq:MTG1}
\Psi_{\mathfrak{d},r,w}^{\otimes (\ell)} \, =\,
\sum_\nu a_\nu u_1^{(\nu)} \otimes \cdots \otimes u_\ell^{(\nu)} \; \in \; \mT^{(\ell)} \, , 
\eeq
where each $u_i^{(\nu)}$ is a level-one state of the form \eqref{eq:DDF1}, and
\bit
\item $\mathfrak{d}$ is the total depth (equal to the sum of the individual
 depths $\sum_{i=1}^\ell \md_i$ in \eqref{eq:MTG1});
\item $r$ is the $\mathfrak{sl}(2)$ representation (given through its dimension);
\item $w$ is half the $H_0 \equiv h_1$ weight of the corresponding state in the given $
\mathfrak{sl}(2)$ representation, {\em i.e.} 
\beq 
H_0 \, \Psi_{\mathfrak{d},r,w}^{\otimes (\ell)} \,=\, 
 2w \Psi_{\mathfrak{d},r,w}^{\otimes (\ell)};
\eeq
\item $\ell$ denotes the level.
\eit

Given a highest weight and a corresponding set of coset Virasoro eigenvalues,
together with the explicit representation of $\mT^{(1)}$ in terms of
DDF states, it is in principle straightforward to obtain the extremal state
of the respective MTG multiplet. We start from a
basis of Fock states \eqref{eq:Fock1} with the level, depth and weight 
according to the highest weight. Imposing the three conditions 
\eqref{eq:cond1}--\eqref{eq:cond3} singles out particular linear combinations
that is the MTG we are looking for. The MTGs are determined up to an overall
normalization.In the following section we will exhibit several examples to
show that the determination of the MTGs can be efficiently implemented and
automated with the present formalism, at least for low levels, in a way 
that reaches substantially beyond known results for $\mF$.

\section{Modules and maximal tensor ground states for \texorpdfstring{$\ell\leq 5$}{ell<=5}}
\label{sec:3}

In this section we give some examples for the MTGs
and spell out the decomposition of the tensor algebra explicitly up to level 5.
Once we have determined the set of MTGs on a given
level-$\ell$ we can obtain any tensor DDF state in $\mT^{(\ell)}$ by 
the joint action of the affine generators $T_m$ 
and the $(\ell-1)$ coset Virasoro operators 
$\cL{k}{m}$ ($2 \le k \le \ell$) with $m\leq -1$.

For simplicity we restrict attention in the remainder to the subalgebra 
$\mTh\subset \mT$, removing those representations which are not relevant in 
relation to $\mF$. This preserves the Virasoro structure, but makes the listing 
of the representations a little more economical.

\subsection{The tensor algebra up to level 5}
\label{sec:3.1}

From \cite{CKMN} we recall that the first three levels of the tensor 
algebra $\mTh$ are given by
\begin{align}\label{eq:T13}
\begin{aligned}
\mTh^{(1)} & \,\equiv \, \mF^{(1)} \,=\, L(\bL_0 + 2\bd) \, , \\[2ex]
\mTh^{(2)} &\,=\, \Vir(\tfrac{1}{2},\tfrac{1}{2}) \otimes L(2\bL_1 + 3\bd) 
\, , \\[2ex]
\mTh^{(3)} &\,=\, \Vir(\tfrac{7}{10},\tfrac{3}{2}) 
\otimes \Vir(\tfrac{1}{2},\tfrac{1}{2}) 
\otimes L(3\bL_0 + 4\bd) \\
&\quad \oplus \Vir(\tfrac{7}{10},\tfrac{1}{10}) 
\otimes \Vir(\tfrac{1}{2},\tfrac{1}{2}) 
\otimes L(\bL_0 + 2\bL_1 + 5\bd)\, .
\end{aligned}
\end{align}
For level 4, the decomposition w.r.t. \eqref{eq:VirAff} gives rise to
six contributions, \cite{CKMN}
\begin{align}
\begin{aligned}
\mTh^{(4)} &\,=\, 
\Vir(\tfrac{4}{5},0) \otimes \Vir(\tfrac{7}{10},\tfrac{3}{2}) 
\otimes \Vir(\tfrac{1}{2},\tfrac{1}{2}) \otimes L(4\bL_0 + 6\bd) \\
&\quad \oplus \Vir(\tfrac{4}{5},\tfrac{2}{3}) 
\otimes \Vir(\tfrac{7}{10},\tfrac{3}{2}) 
\otimes \Vir(\tfrac{1}{2},\tfrac{1}{2}) 
\otimes L(2\bL_0 + 2\bL_1 + 5\bd) \\
&\quad \oplus \Vir(\tfrac{4}{5},3) 
\otimes \Vir(\tfrac{7}{10},\tfrac{3}{2}) 
\otimes \Vir(\tfrac{1}{2},\tfrac{1}{2}) 
\otimes L(4\bL_1 + 2\bd) \\
&\quad \oplus \Vir(\tfrac{4}{5},\tfrac{7}{5}) 
\otimes \Vir(\tfrac{7}{10},\tfrac{1}{10}) 
\otimes \Vir(\tfrac{1}{2},\tfrac{1}{2}) 
\otimes L(4\bL_0 + 6\bd) \\
&\quad \oplus \Vir(\tfrac{4}{5},\tfrac{1}{15}) 
\otimes \Vir(\tfrac{7}{10},\tfrac{1}{10}) 
\otimes \Vir(\tfrac{1}{2},\tfrac{1}{2}) 
\otimes L(2\bL_0 + 2 \bL_1 + 7\bd) \\
&\quad \oplus \Vir(\tfrac{4}{5},\tfrac{2}{5}) 
\otimes \Vir(\tfrac{7}{10},\tfrac{1}{10}) 
\otimes \Vir(\tfrac{1}{2},\tfrac{1}{2}) \otimes L(4\bL_1 + 6\bd) \, .
\end{aligned}
\end{align}
Similarly, for level 5, we find $6\times 3 = 18$ contributions
\begin{align}
\begin{aligned}\label{eq:T5}
\mTh^{(5)} &\,=\, 
\Vir(\tfrac{6}{7},0) \otimes 
\Vir(\tfrac{4}{5},0) \otimes 
\Vir(\tfrac{7}{10},\tfrac{3}{2}) \otimes 
\Vir(\tfrac{1}{2},\tfrac{1}{2}) \otimes L(5\bL_0 + 8\bd) \\
&\quad \oplus \Vir(\tfrac{6}{7},\tfrac{5}{7}) \otimes 
\Vir(\tfrac{4}{5},0) \otimes 
\Vir(\tfrac{7}{10},\tfrac{3}{2}) \otimes 
\Vir(\tfrac{1}{2},\tfrac{1}{2}) \otimes L(3\bL_0 + 2 \bL_1 + 7\bd) \\
&\quad \oplus \Vir(\tfrac{6}{7},\tfrac{22}{7}) \otimes 
\Vir(\tfrac{4}{5},0) \otimes 
\Vir(\tfrac{7}{10},\tfrac{3}{2}) \otimes 
\Vir(\tfrac{1}{2},\tfrac{1}{2}) \otimes L(\bL_0 + 4 \bL_1 + 4\bd) \\
&\quad \oplus \Vir(\tfrac{6}{7},\tfrac{4}{3}) \otimes 
\Vir(\tfrac{4}{5},\tfrac{2}{3}) \otimes 
\Vir(\tfrac{7}{10},\tfrac{3}{2}) \otimes 
\Vir(\tfrac{1}{2},\tfrac{1}{2}) \otimes L(5\bL_0 + 6\bd) \\
&\quad \oplus \Vir(\tfrac{6}{7},\tfrac{1}{21}) \otimes 
\Vir(\tfrac{4}{5},\tfrac{2}{3}) \otimes 
\Vir(\tfrac{7}{10},\tfrac{3}{2}) \otimes 
\Vir(\tfrac{1}{2},\tfrac{1}{2}) \otimes L(3\bL_0 + 2 \bL_1 + 7\bd) \\
&\quad \oplus \Vir(\tfrac{6}{7},\tfrac{10}{21}) \otimes 
\Vir(\tfrac{4}{5},\tfrac{2}{3}) \otimes 
\Vir(\tfrac{7}{10},\tfrac{3}{2}) \otimes 
\Vir(\tfrac{1}{2},\tfrac{1}{2}) \otimes L(\bL_0 + 4 \bL_1 + 6\bd) \\
&\quad \oplus \Vir(\tfrac{6}{7},5) \otimes 
\Vir(\tfrac{4}{5},3) \otimes 
\Vir(\tfrac{7}{10},\tfrac{3}{2}) \otimes 
\Vir(\tfrac{1}{2},\tfrac{1}{2}) \otimes L(5\bL_0) \\
&\quad \oplus \Vir(\tfrac{6}{7},\tfrac{12}{7}) \otimes 
\Vir(\tfrac{4}{5},3) \otimes 
\Vir(\tfrac{7}{10},\tfrac{3}{2}) \otimes 
\Vir(\tfrac{1}{2},\tfrac{1}{2}) \otimes L(3\bL_0 + 2 \bL_1 + 3\bd) \\
&\quad \oplus \Vir(\tfrac{6}{7},\tfrac{1}{7}) \otimes 
\Vir(\tfrac{4}{5},3) \otimes 
\Vir(\tfrac{7}{10},\tfrac{3}{2}) \otimes 
\Vir(\tfrac{1}{2},\tfrac{1}{2}) \otimes L(\bL_0 + 4 \bL_1 + 4\bd) \\
&\quad \oplus \Vir(\tfrac{6}{7},0) \otimes 
\Vir(\tfrac{4}{5},\tfrac{7}{5}) \otimes 
\Vir(\tfrac{7}{10},\tfrac{1}{10}) \otimes 
\Vir(\tfrac{1}{2},\tfrac{1}{2}) \otimes L(5\bL_0 + 8\bd) \\
&\quad \oplus \Vir(\tfrac{6}{7},\tfrac{5}{7}) \otimes 
\Vir(\tfrac{4}{5},\tfrac{7}{5}) \otimes 
\Vir(\tfrac{7}{10},\tfrac{1}{10}) \otimes 
\Vir(\tfrac{1}{2},\tfrac{1}{2}) \otimes L(3\bL_0 + 2 \bL_1 + 7\bd) \\
&\quad \oplus \Vir(\tfrac{6}{7},\tfrac{22}{7}) \otimes 
\Vir(\tfrac{4}{5},\tfrac{7}{5}) \otimes 
\Vir(\tfrac{7}{10},\tfrac{1}{10}) \otimes 
\Vir(\tfrac{1}{2},\tfrac{1}{2}) \otimes L(\bL_0 + 4 \bL_1 + 4\bd) \\
&\quad \oplus \Vir(\tfrac{6}{7},\tfrac{4}{3}) \otimes 
\Vir(\tfrac{4}{5},\tfrac{1}{15}) \otimes 
\Vir(\tfrac{7}{10},\tfrac{1}{10}) \otimes 
\Vir(\tfrac{1}{2},\tfrac{1}{2}) \otimes L(5\bL_0 + 8\bd) \\
&\quad \oplus \Vir(\tfrac{6}{7},\tfrac{1}{21}) \otimes 
\Vir(\tfrac{4}{5},\tfrac{1}{15}) \otimes 
\Vir(\tfrac{7}{10},\tfrac{1}{10}) \otimes 
\Vir(\tfrac{1}{2},\tfrac{1}{2}) \otimes L(3\bL_0 + 2 \bL_1 + 9\bd) \\
&\quad \oplus \Vir(\tfrac{6}{7},\tfrac{10}{21}) \otimes 
\Vir(\tfrac{4}{5},\tfrac{1}{15}) \otimes 
\Vir(\tfrac{7}{10},\tfrac{1}{10}) \otimes 
\Vir(\tfrac{1}{2},\tfrac{1}{2}) \otimes L(\bL_0 + 4 \bL_1 + 8\bd) \\
&\quad \oplus \Vir(\tfrac{6}{7},5) \otimes 
\Vir(\tfrac{4}{5},\tfrac{2}{5}) \otimes 
\Vir(\tfrac{7}{10},\tfrac{1}{10}) \otimes 
\Vir(\tfrac{1}{2},\tfrac{1}{2}) \otimes L(5\bL_0 + 4\bd) \\
&\quad \oplus \Vir(\tfrac{6}{7},\tfrac{12}{7}) \otimes 
\Vir(\tfrac{4}{5},\tfrac{2}{5}) \otimes 
\Vir(\tfrac{7}{10},\tfrac{1}{10}) \otimes 
\Vir(\tfrac{1}{2},\tfrac{1}{2}) \otimes L(3\bL_0 + 2 \bL_1 + 7\bd) \\
&\quad \oplus \Vir(\tfrac{6}{7},\tfrac{1}{7}) \otimes 
\Vir(\tfrac{4}{5},\tfrac{2}{5}) \otimes 
\Vir(\tfrac{7}{10},\tfrac{1}{10}) \otimes 
\Vir(\tfrac{1}{2},\tfrac{1}{2}) \otimes L(\bL_0 + 4 \bL_1 + 8\bd) \, .
\end{aligned}
\end{align}
For $\ell \ge 3$ these results were obtained by multiplying $L(\bL_0 + 2\bd)$
with $\mT^{(\ell-1)}$ and using \eqref{eq:prodL}. In principle, the tensor
algebra can therefore be obtained successively to arbitrary level. The 
expressions for the full tensor products $\mT^{(\ell)}$ are similar but twice as
long. The correctness of the above expressions is ensured by Theorem~4.1 of
\cite{KW1}, but can also be checked by matching characters on the l.h.s. and 
r.h.s. and comparing the associated $q$-series (where $q = e^{-\bd}$), which
when expanded to sufficiently high orders leads to unique answers.
This simple reasoning would be somewhat obscured if we followed the usual 
practice of including fractional powers of $q$ in the definition of
the Virasoro characters, whereas the l.h.s. contains only integer powers of $q$; 
hence the exponents on the r.h.s. must always combine to integers.
In other words, using \cite{KW1} we have gained complete control over
the tensor product spaces $\mTh^{(\ell)}$ and $\mT^{(\ell)}$ for arbirary
levels $\ell$, which can now be worked out with little effort.
Let us also note the multiple appearance from level $\ell=4$
onwards of affine modules with the same weights, but differing in the
$h$ eigenvalues of the accompanying Virasoro modules.

In future work it will be interesting to investigate the character identities
that follow from the above decompositions. For those it is convenient to adopt
a normalization of the characters {\em without} fractional exponents, because 
the l.h.s. of these equalities has no fractional exponents. For instance, for
the Virasoro characters occurring in the above equations we take \cite{Welsh}
\beq\label{eq:VirCh}
\chi^{p,p'}_{r,s} (q) \,\equiv\, {\rm Tr}\; q^{\mL_0 -h^{(\ell)}_{r,s}} \,=\,
\varphi^{-1}(q) \sum_{k\in\mathbb{Z}} \left(
q^{pp'k^2 + k(pr - p's)} - q^{(p'k+r)(pk +s)} \right) \, ,
\eeq
where $p = p^\prime + 1 = \ell+2$ and $\mL_0$ is the appropriate coset Virasoro 
generator. To be sure, fractional exponents are needed to exhibit the modular 
properties of the characters and the associated $\Theta$-functions, but these
properties will now have to be analyzed {\em jointly} 
for the affine and Virasoro characters. We leave this problem to future study.

\subsection{Some explicit MTGs}
\label{sec:3.2}

In this subsection, we give some examples for MTGs.
The states have been computed using the DDF Mathematica package 
\cite{DDFPackage}, which offers a direct implementation of the algorithm
outlined in section~\ref{sec:2.4}. 

Before proceeding let us first clarify the relation between the MTGs
and the maximal Virasoro ground states introduced in \cite{CKMN}.
The latter are elements of the tensor product $\mF^{(1)}\otimes \mF^{(\ell -1)}$,
and thus neither in $\mTh^{(\ell)}$ nor $\mF^{(\ell)}$. According to \cite{CKMN}
the maximal (Virasoro) ground states $\Psi_{\mathfrak{d},r,w}^{(\ell)}$ 
are thus defined by an incomplete conversion of the tensor product via
\beq
\Psi_{\mathfrak{d}+\ell-2,r,w}^{(\ell)} \,=\, 
\left(\bE \otimes \cI_{\ell-1} \cL{\ell-1}{-1} \right) \ldots 
\left( \bE \otimes \ldots \otimes \bE \otimes \cI_2 \cL{2}{-1} \right) 
\Psi_{\mathfrak{d},r,w}^{\otimes(\ell)} \,\in \,
\mathfrak{F}^{(1)} \otimes \; \mF^{(\ell -1)} 
\eeq
up to some irrelevant normalization constant; they thus allow  only for a
realization of the `last' level-$\ell$ coset Virasoro algebra. By contrast, 
the MTGs defined above are more general because they allow for the 
simultaneous action of all coset Virasoros of level $k \le \ell$, whereas 
this is not possible for $\mF^{(1)} \otimes \mF^{(\ell -1)}$. On the latter 
maximal Virasoro  ground states we can thus only act with $\cL{\ell}{m}$, 
which is not enough to generate the full $\mF^{(\ell)}$. Note also that
here we have to insert coset Virasoro operators in order to obtain a 
non-vanishing element of $\mF^{(\ell -1)}$, which forces us to `undo'
the previous conversion by trading $\mF^{(\ell -1)}$
for $\mF^{(1)} \otimes \mF^{(\ell -2)}$. For this
reason the MTGs employed here completely replace the maximal (Virasoro) ground 
states from \cite{CKMN}, and the latter will no longer be used here.

\subsubsection*{Level 1 and level 2}
On level-one there is as yet no action of a Virasoro algebra,
there is only a single affine ground state
\beq
\Psi_{0,\bE,0}^{\otimes(1)} \,=\, \ket{\ba_0^{(1)}} \, , \quad 
\eeq
from which all states in $\mT^{(1)}$ can be reached
by the action of the affine generators. This is the only
ground state that is also an element of $\mF$.
At level 2 there is one MTG that
agrees with the maximal (Virasoro) ground state from \cite{CKMN}, 
which is
\beq
\label{eq:mtg2}
\Psi_{1,\bD,1}^{\otimes(2)} \,=\, \ket{\ba_{1}^{(1)}} \wedge \ket{\ba_0^{(1)}} \, .
\eeq
This MTG, and all the others, are {\em virtual}, {\em i.e.} they are mapped
to zero when the tensor products are converted to Lie algebra commutators.

\subsubsection*{Level 3}
At level 3 we find two MTGs. The MTG corresponding to
\beq
\Vir(\tfrac{7}{10},\tfrac{3}{2}) \otimes 
\Vir(\tfrac{1}{2},\tfrac{1}{2}) \otimes L(3\bL_0 + 4\bd)
\eeq
is
\begin{align}\label{eq:mtg31}
\begin{aligned}
\Psi_{2,\boldsymbol{1},0}^{\otimes \, (3)} 
&\,=\, \ket{\ba_{0}^{(1)}} \otimes \ket{\ba_{0}^{(1)}} \wedge 
\Ae_{-1} \Ae_{-1} \ket{\ba_{0}^{(1)}} \\
&\quad \,+\, \ket{\ba_{-1}^{(1)}} \otimes \Psi_{1,\mathbf{3},1}^{\otimes \, (2)} 
 \,+\, 2 \ \Ae_{-1} \ket{\ba_0^{(1)}} \otimes \Psi_{1,\mathbf{3},0}^{\otimes \, (2)} 
 \,+\, \ket{\ba_{1}^{(1)}} \otimes \Psi_{1,\mathbf{3},-1}^{\otimes \, (2)} \, .
\end{aligned}
\end{align}
The states $\Psi_{1,\mathbf{3},0}^{\otimes \, (2)}$, respectively 
$\Psi_{1,\mathbf{3},-1}^{\otimes \, (2)}$, are obtained from 
$\Psi_{1,\mathbf{3},1}^{\otimes \, (2)}$ by the application of $\frac{1}{2} \, {}^{[2]}\!F_0$, 
respectively $-\frac{1}{2} \, {}^{[2]}\!F_0 \, {}^{[2]}\!F_0$ .
The MTG corresponding to
\beq
\Vir(\tfrac{7}{10},\tfrac{1}{10}) \otimes 
\Vir(\tfrac{1}{2},\tfrac{1}{2}) \otimes L(\bL_0 + 2\bL_1 + 5\bd)
\eeq
is
\beq\label{eq:mtg32}
\Psi_{1,\boldsymbol{3},1}^{\otimes \, (3)} 
\,=\, \ket{\ba_0^{(1)}} \otimes \Psi_{1,\boldsymbol{3},1}^{\otimes \, (2)} 
\,=\, \ket{\ba_0^{(1)}} \otimes \ket{\ba_0^{(1)}} \wedge \ket{\ba_1^{(1)}} \, .
\eeq

\subsubsection*{Level 4}
On level 4 we have computed all six MTGs. Here we present only
five of these expressions; the sixth MTG is known but too lengthy
and shall be omitted here.
The singlet MTG corresponding to
\beq
\Vir(\tfrac{4}{5},0) \otimes \Vir(\tfrac{7}{10},\tfrac{3}{2}) 
\otimes \Vir(\tfrac{1}{2},\tfrac{1}{2}) \otimes L(4\bL_0 + 6\bd) \, .
\eeq
is
\beq\label{eq:mtg41}
\Psi_{2,\boldsymbol{1},0}^{\otimes \, (4)} \,=\, \ket{\ba_0^{(1)}} 
\otimes \Psi_{2,\boldsymbol{1},0}^{\otimes \, (3)} \, . 
\eeq
The triplet MTG corresponding to
\beq
\Vir(\tfrac{4}{5},\tfrac{2}{3}) 
\otimes \Vir(\tfrac{7}{10},\tfrac{3}{2}) 
\otimes \Vir(\tfrac{1}{2},\tfrac{1}{2}) 
\otimes L(2\bL_0 + 2\bL_1 + 5\bd)
\eeq
is
\begin{align}\label{eq:mtg42}
\begin{aligned}
\Psi_{3,\mathbf{3},1}^{\otimes \, (4)} 
&\,=\, \sqrt{2} \ket{\ba_{0}^{(1)}} \otimes \ket{\ba_{0}^{(1)}} 
\otimes \ket{\ba_{0}^{(1)}} \wedge \Ae_{-2} \ket{\ba_{1}^{(1)}} \\
&\quad \,-\, \ket{\ba_{0}^{(1)}} \otimes \ket{\ba_{0}^{(1)}} 
\otimes \ket{\ba_{0}^{(1)}} \wedge \Ae_{-1} \Ae_{-1} \ket{\ba_{1}^{(1)}} \\
&\quad \,-\, \ket{\ba_{0}^{(1)}} \otimes \ket{\ba_{0}^{(1)}} 
\otimes \Ae_{-1} \Ae_{-1} \ket{\ba_{0}^{(1)}} \wedge \ket{\ba_{1}^{(1)}} \\
&\quad \,-\, \ket{\ba_{0}^{(1)}} \otimes \ket{\ba_{1}^{(1)}} 
\otimes \ket{\ba_{0}^{(1)}} \wedge \Ae_{-1} \Ae_{-1} \ket{\ba_{1}^{(1)}} \\
&\quad \,+\, 2 \ket{\ba_{0}^{(1)}} \otimes \Ae_{-1} \ket{\ba_{0}^{(1)}} 
\otimes \ket{\ba_{0}^{(1)}} \wedge \Ae_{-1} \ket{\ba_{1}^{(1)}} \\
&\quad \,-\, \sqrt{2} \ket{\ba_{0}^{(1)}} \otimes \ket{\ba_{1}^{(1)}} 
\otimes \ket{\ba_{0}^{(1)}} \wedge \Ae_{-2} \ket{\ba_{1}^{(1)}} \\
&\quad \,+\, 2 \ket{\ba_{0}^{(1)}} \otimes \ket{\ba_{1}^{(1)}} 
\otimes \ket{\ba_{1}^{(1)}} \wedge \ket{\ba_{-1}^{(1)}} \\
&\quad \,-\, 2 \ket{\ba_{0}^{(1)}} \otimes 
\Ae_{-1} \ket{\ba_{0}^{(1)}} \otimes \ket{\ba_{1}^{(1)}} 
\wedge \Ae_{-1} \ket{\ba_{0}^{(1)}} \\
&\quad \,-\, 2 \ket{\ba_{0}^{(1)}} \otimes \Ae_{-1} \Ae_{-1} \ket{\ba_{0}^{(1)}} 
\otimes \ket{\ba_{0}^{(1)}} \wedge \ket{\ba_{1}^{(1)}} \\
&\quad \,-\, 2 \ket{\ba_{0}^{(1)}} \otimes \Ae_{-1} \ket{\ba_{1}^{(1)}} 
\otimes \Ae_{-1} \ket{\ba_{0}^{(1)}} \wedge \ket{\ba_{0}^{(1)}} \\
&\quad \,+\,\sqrt{2} \ket{\ba_{0}^{(1)}} \otimes \Ae_{-2} \ket{\ba_{0}^{(1)}} 
\otimes \ket{\ba_{1}^{(1)}} \wedge \ket{\ba_{0}^{(1)}} \\
&\quad \,-\, 3 \ket{\ba_{1}^{(1)}} 
\otimes \Psi_{2,\boldsymbol{1},0}^{\otimes \, (3)} \, .
\end{aligned}
\end{align}
The fiveplet MTG corresponding to
\beq
\Vir(\tfrac{4}{5},3) 
\otimes \Vir(\tfrac{7}{10},\tfrac{3}{2}) 
\otimes \Vir(\tfrac{1}{2},\tfrac{1}{2}) 
\otimes L(4\bL_1 + 2\bd)
\eeq
is known, but we omit the result because it is simply
too long. The singlet MTG corresponding to
\beq
\Vir(\tfrac{4}{5},\tfrac{7}{5}) \otimes 
\Vir(\tfrac{7}{10},\tfrac{1}{10}) \otimes 
\Vir(\tfrac{1}{2},\tfrac{1}{2}) \otimes L(4\bL_0 + 6\bd) 
\eeq
is
\beq\label{eq:mtg43}
\bar\Psi_{2,\boldsymbol{1},0}^{\otimes \, (4)} 
\,=\, \ket{\ba_0^{(1)}} \otimes \bar\Psi_{2,\boldsymbol{1},0}^{\otimes \, (3)} 
\,-\, 7 \ket{\ba_{-1}^{(1)}} 
\otimes \Psi_{1,\mathbf{3},1}^{\otimes \, (3)} 
\,-\, 14 \, \Ae_{-1} \ket{\ba_{0}^{(1)}} 
\otimes \Psi_{1,\mathbf{3},0}^{\otimes \, (3)} 
\,-\, 7 \ket{\ba_{1}^{(1)}} 
\otimes \Psi_{1,\mathbf{3},-1}^{\otimes \, (3)} \, . 
\eeq
The triplet MTG corresponding to
\beq
\Vir(\tfrac{4}{5},\tfrac{1}{15}) \otimes 
\Vir(\tfrac{7}{10},\tfrac{1}{10}) \otimes 
\Vir(\tfrac{1}{2},\tfrac{1}{2}) \otimes L(2\bL_0 + 2 \bL_1 + 7\bd)
\eeq
is
\beq\label{eq:mtg44}
\Psi_{1,\boldsymbol{3},1}^{\otimes \, (4)} \,=\, \ket{\ba_0^{(1)}} \otimes 
\Psi_{1,\boldsymbol{3},1}^{\otimes \, (3)} \,=\, \ket{\ba_0^{(1)}} \otimes 
\ket{\ba_0^{(1)}} \otimes \ket{\ba_0^{(1)}} \wedge \ket{\ba_1^{(1)}} \, .
\eeq
The fiveplet MTG corresponding to
\beq
\Vir(\tfrac{4}{5},\tfrac{2}{5}) \otimes 
\Vir(\tfrac{7}{10},\tfrac{1}{10}) \otimes 
\Vir(\tfrac{1}{2},\tfrac{1}{2}) \otimes L(4\bL_1 + 6\bd)
\eeq
is
\beq\label{eq:mtg45}
\Psi_{2,\mathbf{5}.2}^{\otimes \, (4)} 
\,=\, \ket{\ba_0^{(1)}} \otimes \ket{\ba_1^{(1)}} \otimes 
\ket{\ba_0^{(1)}} \wedge \ket{\ba_1^{(1)}} 
\,-\, \ket{\ba_0^{(1)}} \otimes \ket{\ba_0^{(1)}} \otimes 
\ket{\ba_0^{(1)}} \wedge \ket{\ba_1^{(1)}} \, .
\eeq

\subsubsection*{Level 5}
Of the 18 MTG at level 5 we have computed 8 which we present in the following.
The MTG corresponding to 
\beq
\Vir(\tfrac{6}{7},0) \otimes \Vir(\tfrac{4}{5},0) \otimes 
\Vir(\tfrac{7}{10},\tfrac{3}{2}) \otimes 
\Vir(\tfrac{1}{2},\tfrac{1}{2}) \otimes L(5\bL_0 + 8\bd) 
\eeq
is
\begin{align}
\begin{aligned}
\Psi_{2,\mathbf{1},0}^{\otimes (5)} &\,=\,
\,-\,\ket{\ba_0^{(1)}} \otimes \ket{\ba_0^{(1)}} \otimes \Ae_{-1} \ket{\ba_0^{(1)}} 
\otimes \Ae_{-1} \ket{\ba_0^{(1)}} \wedge \ket{\ba_0^{(1)}} \\
&\quad \,+\, \frac{1}{2} \ket{\ba_0^{(1)}} \otimes \ket{\ba_0^{(1)}} \otimes 
\ket{\ba_0^{(1)}} \otimes \Ae_{-1} \Ae_{-1} \ket{\ba_0^{(1)}} \wedge 
\ket{\ba_0^{(1)}} \\
&\quad \,-\, \ket{\ba_0^{(1)}} \otimes \ket{\ba_0^{(1)}} \otimes 
\ket{\ba_{-1}^{(1)}} \otimes \ket{\ba_0^{(1)}} \wedge \ket{\ba_1^{(1)}} \\
&\quad \,-\, \ket{\ba_0^{(1)}} \otimes \ket{\ba_0^{(1)}} \otimes \ket{\ba_1^{(1)}} 
\otimes \ket{\ba_0^{(1)}} \wedge \ket{\ba_{-1}^{(1)}} \, .
\end{aligned}
\end{align}
The MTG corresponding to 
\beq
\Vir(\tfrac{6}{7},\tfrac{5}{7}) \otimes \Vir(\tfrac{4}{5},0) 
\otimes \Vir(\tfrac{7}{10},\tfrac{3}{2}) 
\otimes \Vir(\tfrac{1}{2},\tfrac{1}{2}) 
\otimes L(3\bL_0 + 2 \bL_1 + 7\bd) 
\eeq
is
\begin{align}
\begin{aligned}
\Psi_{3,\mathbf{3},1}^{\otimes (5)} &\,=\,
\sqrt{2} \ket{\ba_0^{(1)}} \otimes \ket{\ba_0^{(1)}} \otimes \Ae_{-2} \ket{\ba_0^{(1)}} 
\otimes \ket{\ba_0^{(1)}} \wedge \ket{\ba_1^{(1)}} \\
&\quad \,-\, 2\ket{\ba_0^{(1)}} \otimes \ket{\ba_0^{(1)}} 
\otimes \Ae_{-1} \ket{\ba_0^{(1)}} 
\otimes \Ae_{-1} \ket{\ba_0^{(1)}} \wedge \ket{\ba_1^{(1)}} \\
&\quad \,-\, 2 \sqrt{2} \ket{\ba_0^{(1)}} 
\otimes \ket{\ba_0^{(1)}} \otimes \Ae_{-1} \ket{\ba_0^{(1)}} 
\otimes \ket{\ba_0^{(1)}} \wedge \Ae_{-1} \ket{\ba_1^{(1)}} \\
&\quad \,+\, 2 \ket{\ba_0^{(1)}} \otimes \ket{\ba_0^{(1)}} 
\otimes \Ae_{-1} \ket{\ba_1^{(1)}} \otimes \Ae_{-1} \ket{\ba_0^{(1)}} 
\wedge \ket{\ba_0^{(1)}} \\
&\quad \,+\, 2 \ket{\ba_0^{(1)}} \otimes \ket{\ba_0^{(1)}} 
\otimes \Ae_{-1} \Ae_{-1} \ket{\ba_0^{(1)}} \otimes \ket{\ba_0^{(1)}} 
\wedge \ket{\ba_1^{(1)}} \\
&\quad \,+\, \sqrt{2} \ket{\ba_0^{(1)}} \otimes \ket{\ba_0^{(1)}} \otimes \ket{\ba_0^{(1)}} 
\otimes \Ae_{-2} \ket{\ba_1^{(1)}} \wedge \ket{\ba_0^{(1)}} \\
&\quad \,+\, \ket{\ba_0^{(1)}} \otimes \ket{\ba_0^{(1)}} 
\otimes \ket{\ba_0^{(1)}} \otimes \Ae_{-1} \Ae_{-1} \ket{\ba_0^{(1)}} 
\wedge \ket{\ba_1^{(1)}} \\
&\quad \,+\, \ket{\ba_0^{(1)}} \otimes \ket{\ba_0^{(1)}} 
\otimes \ket{\ba_0^{(1)}} \otimes \ket{\ba_0^{(1)}} 
\wedge \Ae_{-1} \Ae_{-1} \ket{\ba_1^{(1)}} \\
&\quad \,-\, \sqrt{2} \ket{\ba_0^{(1)}} \otimes \ket{\ba_0^{(1)}} \otimes \ket{\ba_1^{(1)}} 
\otimes \Ae_{-2} \ket{\ba_0^{(1)}} \wedge \ket{\ba_0^{(1)}} \\
&\quad \,-\, \ket{\ba_0^{(1)}} \otimes \ket{\ba_0^{(1)}} 
\otimes \ket{\ba_1^{(1)}} \otimes \Ae_{-1} \Ae_{-1} \ket{\ba_0^{(1)}} 
\wedge \ket{\ba_0^{(1)}} \\
&\quad \,+\, 2\ket{\ba_0^{(1)}} \otimes \ket{\ba_0^{(1)}} 
\otimes \ket{\ba_1^{(1)}} \otimes \ket{\ba_{-1}^{(1)}} 
\wedge \ket{\ba_1^{(1)}} \\
&\quad \,-\, 2\ket{\ba_0^{(1)}} \otimes \ket{\ba_1^{(1)}} 
\otimes \Ae_{-1} \ket{\ba_0^{(1)}} \otimes \Ae_{-1} \ket{\ba_0^{(1)}} 
\wedge \ket{\ba_0^{(1)}} \\
&\quad \,+\, \ket{\ba_0^{(1)}} \otimes \ket{\ba_1^{(1)}} 
\otimes \ket{\ba_0^{(1)}} \otimes \Ae_{-1} \Ae_{-1} \ket{\ba_0^{(1)}} 
\wedge \ket{\ba_0^{(1)}} \\
&\quad \,-\, 2\ket{\ba_0^{(1)}} \otimes \ket{\ba_1^{(1)}} 
\otimes \ket{\ba_{-1}^{(1)}} \otimes \ket{\ba_0^{(1)}} 
\wedge \ket{\ba_1^{(1)}} \\
&\quad \,-\, 2\ket{\ba_0^{(1)}} \otimes \ket{\ba_1^{(1)}} 
\otimes \ket{\ba_1^{(1)}} \otimes \ket{\ba_0^{(1)}} 
\wedge \ket{\ba_{-1}^{(1)}} \\
&\quad \,+\, 8 \ket{\ba_1^{(1)}} \otimes \ket{\ba_0^{(1)}} 
\otimes \Ae_{-1} \ket{\ba_0^{(1)}} \otimes \Ae_{-1} \ket{\ba_0^{(1)}} 
\wedge \ket{\ba_0^{(1)}} \\
&\quad \,-\, 4 \ket{\ba_1^{(1)}} \otimes \ket{\ba_0^{(1)}} 
\otimes \ket{\ba_0^{(1)}} \otimes \Ae_{-1} \Ae_{-1} \ket{\ba_0^{(1)}} 
\wedge \ket{\ba_0^{(1)}} \\
&\quad \,+\, 8\ket{\ba_1^{(1)}} \otimes \ket{\ba_0^{(1)}} 
\otimes \ket{\ba_{-1}^{(1)}} \otimes \ket{\ba_0^{(1)}} 
\wedge \ket{\ba_1^{(1)}} \\
&\quad \,+\, 8 \ket{\ba_1^{(1)}} \otimes \ket{\ba_0^{(1)}} 
\otimes \ket{\ba_1^{(1)}} \otimes \ket{\ba_0^{(1)}} \wedge \ket{\ba_{-1}^{(1)}} \, .
\end{aligned}
\end{align}
The MTG corresponding to 
\beq
\Vir(\tfrac{6}{7},\tfrac{1}{21}) 
\otimes \Vir(\tfrac{4}{5},\tfrac{2}{3}) 
\otimes \Vir(\tfrac{7}{10},\tfrac{3}{2}) 
\otimes \Vir(\tfrac{1}{2},\tfrac{1}{2}) 
\otimes L(3\bL_0 + 2 \bL_1 + 7\bd) 
\eeq
is
\begin{align}
\begin{aligned}
\tilde{\Psi}_{3,\mathbf{3},1}^{\otimes (5)} &\,=\,
\sqrt{2} \ket{\ba_{0}^{(1)}} \otimes \ket{\ba_{0}^{(1)}} 
\otimes \Ae_{-2} \ket{\ba_{0}^{(1)}} \otimes \ket{\ba_{0}^{(1)}} 
\wedge \ket{\ba_1^{(1)}} \\
&\quad \,-\, 2 \ket{\ba_{0}^{(1)}} \otimes \ket{\ba_{0}^{(1)}} 
\otimes \Ae_{-1} \ket{\ba_{0}^{(1)}} \otimes \Ae_{-1} \ket{\ba_{0}^{(1)}} 
\wedge \ket{\ba_1^{(1)}} \\ 
&\quad \,-\, 2 \ket{\ba_{0}^{(1)}} \otimes \ket{\ba_{0}^{(1)}} 
\otimes \Ae_{-1} \ket{\ba_{0}^{(1)}} \otimes \ket{\ba_{0}^{(1)}} 
\wedge \Ae_{-1} \ket{\ba_1^{(1)}} \\ 
&\quad \,+\,2 \ket{\ba_{0}^{(1)}} \otimes \ket{\ba_{0}^{(1)}} 
\otimes \Ae_{-1} \ket{\ba_1^{(1)}} \otimes \Ae_{-1} \ket{\ba_{0}^{(1)}} 
\wedge \ket{\ba_{0}^{(1)}} \\ 
&\quad \,+\, 2 \ket{\ba_{0}^{(1)}} \otimes \ket{\ba_{0}^{(1)}} 
\otimes \Ae_{-1} \Ae_{-1} \ket{\ba_{0}^{(1)}} \otimes \ket{\ba_{0}^{(1)}} 
\wedge \ket{\ba_1^{(1)}} \\ 
&\quad \,+\, \sqrt{2} \ket{\ba_{0}^{(1)}} \otimes \ket{\ba_{0}^{(1)}} 
\otimes \ket{\ba_{0}^{(1)}} \otimes \Ae_{-2} \ket{\ba_1^{(1)}} 
\wedge \ket{\ba_{0}^{(1)}} \\ 
&\quad \,+\, \ket{\ba_{0}^{(1)}} \otimes \ket{\ba_{0}^{(1)}} 
\otimes \ket{\ba_{0}^{(1)}} \otimes \Ae_{-1} \Ae_{-1} \ket{\ba_{0}^{(1)}} 
\wedge \ket{\ba_1^{(1)}} \\ 
&\quad \,+\, \ket{\ba_{0}^{(1)}} \otimes \ket{\ba_{0}^{(1)}} 
\otimes \ket{\ba_{0}^{(1)}} \otimes \ket{\ba_{0}^{(1)}} 
\wedge \Ae_{-1} \Ae_{-1} \ket{\ba_1^{(1)}} \\ 
&\quad \,-\, \sqrt{2} \ket{\ba_{0}^{(1)}} \otimes \ket{\ba_{0}^{(1)}} 
\otimes \ket{\ba_1^{(1)}} \otimes \Ae_{-2} \ket{\ba_{0}^{(1)}} 
\wedge \ket{\ba_{0}^{(1)}} \\ 
&\quad \,-\, \ket{\ba_{0}^{(1)}} \otimes \ket{\ba_{0}^{(1)}} 
\otimes \ket{\ba_1^{(1)}} \otimes \Ae_{-1} \Ae_{-1} \ket{\ba_{0}^{(1)}} 
\wedge \ket{\ba_{0}^{(1)}} \\ 
&\quad \,+\,2 \ket{\ba_{0}^{(1)}} \otimes \ket{\ba_{0}^{(1)}} 
\otimes \ket{\ba_1^{(1)}} \otimes \ket{\ba_{-1}^{(1)}} 
\wedge \ket{\ba_1^{(1)}} \\ 
&\quad \,+\, 6 \ket{\ba_{0}^{(1)}} \otimes \ket{\ba_1^{(1)}} 
\otimes \Ae_{-1} \ket{\ba_{0}^{(1)}} \otimes \Ae_{-1} \ket{\ba_{0}^{(1)}} 
\wedge \ket{\ba_{0}^{(1)}} \\ 
&\quad \,-\, 3 \ket{\ba_{0}^{(1)}} \otimes \ket{\ba_1^{(1)}} 
\otimes \ket{\ba_{0}^{(1)}} \otimes \Ae_{-1} \Ae_{-1} \ket{\ba_{0}^{(1)}} 
\wedge \ket{\ba_{0}^{(1)}} \\ 
&\quad \,+\, 6 \ket{\ba_{0}^{(1)}} \otimes \ket{\ba_1^{(1)}} 
\otimes \ket{\ba_{-1}^{(1)}} \otimes \ket{\ba_{0}^{(1)}} 
\wedge \ket{\ba_1^{(1)}} \\ 
&\quad \,+\, 6 \ket{\ba_{0}^{(1)}} \otimes \ket{\ba_1^{(1)}} 
\otimes \ket{\ba_1^{(1)}} \otimes \ket{\ba_{0}^{(1)}} 
\wedge \ket{\ba_{-1}^{(1)}} \, .
\end{aligned}
\end{align}
The MTG corresponding to 
\beq
\Vir(\tfrac{6}{7},0) \otimes \Vir(\tfrac{4}{5},\tfrac{7}{5}) 
\otimes \Vir(\tfrac{7}{10},\tfrac{1}{10}) 
\otimes \Vir(\tfrac{1}{2},\tfrac{1}{2}) 
\otimes L(5\bL_0 + 8\bd) 
\eeq
is
\begin{align}
\begin{aligned}
\tilde{\Psi}_{2,\mathbf{1},0}^{\otimes (5)} &\,=\,
7\ket{\ba_{0}^{(1)}} \otimes \Ae_{-1} \ket{\ba_{0}^{(1)}} 
\otimes \ket{\ba_{0}^{(1)}} \otimes \Ae_{-1} \ket{\ba_{0}^{(1)}} 
\wedge \ket{\ba_{0}^{(1)}} \\
&\quad \,-\, \ket{\ba_{0}^{(1)}} \otimes \ket{\ba_{0}^{(1)}} 
\otimes \Ae_{-1} \ket{\ba_{0}^{(1)}} \otimes \Ae_{-1} \ket{\ba_{0}^{(1)}} 
\wedge \ket{\ba_{0}^{(1)}} \\
&\quad \,-\, 3\ket{\ba_{0}^{(1)}} \otimes \ket{\ba_{0}^{(1)}} 
\otimes \ket{\ba_{0}^{(1)}} \otimes \Ae_{-1} \Ae_{-1} \ket{\ba_{0}^{(1)}} 
\wedge \ket{\ba_{0}^{(1)}} \\
&\quad \,-\, \ket{\ba_{0}^{(1)}} \otimes \ket{\ba_{0}^{(1)}} 
\otimes \ket{\ba_{-1}^{(1)}} \otimes \ket{\ba_{0}^{(1)}} 
\wedge \ket{\ba_1^{(1)}} \\
&\quad \,-\, \ket{\ba_{0}^{(1)}} \otimes \ket{\ba_{0}^{(1)}} 
\otimes \ket{\ba_1^{(1)}} \otimes \ket{\ba_{0}^{(1)}} 
\wedge \ket{\ba_{-1}^{(1)}} \\
&\quad \,+\, 7\ket{\ba_{0}^{(1)}} \otimes \ket{\ba_{-1}^{(1)}} 
\otimes \ket{\ba_{0}^{(1)}} \otimes \ket{\ba_{0}^{(1)}} 
\wedge \ket{\ba_1^{(1)}} \\
&\quad \,+\, 7\ket{\ba_{0}^{(1)}} \otimes \ket{\ba_1^{(1)}} 
\otimes \ket{\ba_{0}^{(1)}} \otimes \ket{\ba_{0}^{(1)}} 
\wedge \ket{\ba_{-1}^{(1)}} \, .
\end{aligned}
\end{align}
The MTG corresponding to 
\beq
\Vir(\tfrac{6}{7},\tfrac{4}{3}) 
\otimes \Vir(\tfrac{4}{5},\tfrac{1}{15}) 
\otimes \Vir(\tfrac{7}{10},\tfrac{1}{10}) 
\otimes \Vir(\tfrac{1}{2},\tfrac{1}{2}) 
\otimes L(5\bL_0 + 8\bd) 
\eeq
is
\begin{align}
\begin{aligned}
\tilde{\tilde{\Psi}}_{2,\mathbf{1},0}^{\otimes (5)} &\,=\,
8\Ae_{-1} \ket{\ba_{0}^{(1)}} \otimes \ket{\ba_{0}^{(1)}} 
\otimes \ket{\ba_{0}^{(1)}} \otimes \Ae_{-1} \ket{\ba_{0}^{(1)}} 
\wedge \ket{\ba_{0}^{(1)}} \\ 
&\quad \,-\, \ket{\ba_{0}^{(1)}} \otimes \Ae_{-1} \ket{\ba_{0}^{(1)}} 
\otimes \ket{\ba_{0}^{(1)}} \otimes \Ae_{-1} \ket{\ba_{0}^{(1)}} 
\wedge \ket{\ba_{0}^{(1)}} \\ 
&\quad \,-\, \ket{\ba_{0}^{(1)}} \otimes \ket{\ba_{0}^{(1)}} 
\otimes \Ae_{-1} \ket{\ba_{0}^{(1)}} \otimes \Ae_{-1} \ket{\ba_{0}^{(1)}} 
\wedge \ket{\ba_{0}^{(1)}} \\ 
&\quad \,-\, 3\ket{\ba_{0}^{(1)}} \otimes \ket{\ba_{0}^{(1)}} 
\otimes \ket{\ba_{0}^{(1)}} \otimes \Ae_{-1} \Ae_{-1} \ket{\ba_{0}^{(1)}} 
\wedge \ket{\ba_{0}^{(1)}} \\ 
&\quad \,-\, \ket{\ba_{0}^{(1)}} \otimes \ket{\ba_{0}^{(1)}} 
\otimes \ket{\ba_{-1}^{(1)}} \otimes \ket{\ba_{0}^{(1)}} 
\wedge \ket{\ba_1^{(1)}} \\ 
&\quad \,-\, \ket{\ba_{0}^{(1)}} \otimes \ket{\ba_{0}^{(1)}} 
\otimes \ket{\ba_1^{(1)}} \otimes \ket{\ba_{0}^{(1)}} 
\wedge \ket{\ba_{-1}^{(1)}} \\ 
&\quad \,-\, \ket{\ba_{0}^{(1)}} \otimes \ket{\ba_{-1}^{(1)}} 
\otimes \ket{\ba_{0}^{(1)}} \otimes \ket{\ba_{0}^{(1)}} 
\wedge \ket{\ba_1^{(1)}} \\ 
&\quad \,-\, \ket{\ba_{0}^{(1)}} \otimes \ket{\ba_1^{(1)}} 
\otimes \ket{\ba_{0}^{(1)}} \otimes \ket{\ba_{0}^{(1)}} 
\wedge \ket{\ba_{-1}^{(1)}} \\ 
&\quad \,+\, 8\ket{\ba_{-1}^{(1)}} \otimes \ket{\ba_{0}^{(1)}} 
\otimes \ket{\ba_{0}^{(1)}} \otimes \ket{\ba_{0}^{(1)}} 
\wedge \ket{\ba_1^{(1)}} \\ 
&\quad \,+\, 8\ket{\ba_1^{(1)}} \otimes \ket{\ba_{0}^{(1)}} 
\otimes \ket{\ba_{0}^{(1)}} \otimes \ket{\ba_{0}^{(1)}} 
\wedge \ket{\ba_{-1}^{(1)}} \, .
\end{aligned}
\end{align}
The MTG corresponding to 
\beq
\Vir(\tfrac{6}{7},\tfrac{1}{21}) 
\otimes \Vir(\tfrac{4}{5},\tfrac{1}{15}) 
\otimes \Vir(\tfrac{7}{10},\tfrac{1}{10}) 
\otimes \Vir(\tfrac{1}{2},\tfrac{1}{2}) 
\otimes L(3\bL_0 + 2 \bL_1 + 9\bd) 
\eeq
is
\beq
\Psi_{1,\mathbf{3},1}^{\otimes (5)} \,=\, \ket{\ba_{0}^{(1)}} 
\otimes \ket{\ba_{0}^{(1)}} \otimes \ket{\ba_{0}^{(1)}} 
\otimes \ket{\ba_1^{(1)}} \wedge \ket{\ba_{0}^{(1)}} \, .
\eeq
The MTG corresponding to 
\beq
\Vir(\tfrac{6}{7},\tfrac{10}{21}) 
\otimes \Vir(\tfrac{4}{5},\tfrac{1}{15}) 
\otimes \Vir(\tfrac{7}{10},\tfrac{1}{10}) 
\otimes \Vir(\tfrac{1}{2},\tfrac{1}{2}) 
\otimes L(\bL_0 + 4 \bL_1 \ + 8\bd) 
\eeq
is
\begin{align}
\begin{aligned}
\Psi_{2,\mathbf{5},2}^{\otimes (5)} &\,=\,
\frac{1}{2} \ket{\ba_{0}^{(1)}} \otimes \ket{\ba_{0}^{(1)}} 
\otimes \ket{\ba_1^{(1)}} \otimes \ket{\ba_{0}^{(1)}} 
\wedge \ket{\ba_1^{(1)}} \\ 
&\quad \,+\, \frac{1}{2} \ket{\ba_{0}^{(1)}} \otimes \ket{\ba_1^{(1)}} 
\otimes \ket{\ba_{0}^{(1)}} \otimes \ket{\ba_{0}^{(1)}} 
\wedge \ket{\ba_1^{(1)}} \\ 
&\quad \,-\, \ket{\ba_1^{(1)}} \otimes \ket{\ba_{0}^{(1)}} 
\otimes \ket{\ba_{0}^{(1)}} \otimes \ket{\ba_{0}^{(1)}} 
\wedge \ket{\ba_1^{(1)}} \, .
\end{aligned}
\end{align}
The MTG corresponding to 
\beq
\Vir(\tfrac{6}{7},\tfrac{1}{7}) 
\otimes \Vir(\tfrac{4}{5},\tfrac{2}{5}) 
\otimes \Vir(\tfrac{7}{10},\tfrac{1}{10}) 
\otimes \Vir(\tfrac{1}{2},\tfrac{1}{2}) 
\otimes L(\bL_0 + 4 \bL_1 \ + 8\bd) 
\eeq
is
\begin{align}
\begin{aligned}
\tilde{\Psi}_{2,\mathbf{5},2}^{\otimes (5)} 
&\,=\, \ket{\ba_{0}^{(1)}} \otimes \ket{\ba_{0}^{(1)}} 
\otimes \ket{\ba_1^{(1)}} \otimes \ket{\ba_{0}^{(1)}} 
\wedge \ket{\ba_1^{(1)}} \\ 
&\quad \,-\, \ket{\ba_{0}^{(1)}} \otimes \ket{\ba_1^{(1)}} 
\otimes \ket{\ba_{0}^{(1)}} \otimes \ket{\ba_{0}^{(1)}} 
\wedge \ket{\ba_1^{(1)}} \, .
\end{aligned}
\end{align}

\section{Mapping the tensor algebra to the hyperbolic Lie algebra \texorpdfstring{$\mF$}{F}}
\label{sec:4}

Once we have determined the tensor DDF states in $\mTh^{(\ell)}$ we must map
them to the KMA level-$\ell$ sector $\mF^{(\ell)}$ by turning tensor products into
multi-commutators. For this purpose we define a generalization of the 
map \eqref{eq:Iell} by 
$\cJ_\ell \,:\, \mT^{(\ell)}\rightarrow\, \mF^{(\ell)}$ via
\beq
\mathcal{J}_\ell \,\coloneqq\,
\mathcal{I}_\ell \Big( ( \bE \otimes \mathcal{I}_{\ell-1}) \ldots 
\Big( (\underbrace{\bE \otimes \ldots \otimes \bE}_{(\ell-2) \text{ times}} 
\otimes \mathcal{I}_2) (u_1 \otimes \ldots \otimes u_\ell) \ldots \Big) \Big) 
\eeq
with $u_1, \dots , u_\ell \in \mathfrak{T}^{(1)}$; more concretely,
\beq\label{eq:cJell}
\cJ_\ell \big(u_1 \otimes \cdots \otimes u_\ell\big) \,\coloneqq\,
\big[ u_1 , \big[ u_2 , \dots \big[ u_{\ell -1}, u_\ell \big] \cdots\big]\big]
\,\in\; \mF^{(\ell)} \, .
\eeq 
Importantly, the map $\cJ_\ell$ commutes with the affine action, that is,
\beq\label{cJT}
\big[ \cJ_\ell \,,\, T_m \big] \,=\, 0
\eeq
since distributivity holds for both tensor products and commutators.
Therefore the affine modules are not affected by the transition
from $\mTh$ to $\mF$. The analog of \eqref{eq:FKerI} reads
\beq\label{eq:FKerJ}
\mF^{(\ell)} \,=\, \mTh^{(\ell)} \Big/ {\rm Ker}\,\cJ_\ell \, .
\eeq
As a consequence of the Jacobi identities and the
Serre relations the kernel of $\cJ_\ell$
is always non-trivial, as can already be seen with the level-2 example
$u_1 = e_{-1} \, , \, u_2 = [e_{-1}, e_0]$ with $e_i$ the Chevalley
generators of $\mF$, for which
\beq
\cJ_2 \Big(e_{-1} \wedge [e_{-1},e_0]\Big) \,=\, 
\big[ e_{-1} , [e_{-1},e_0]\big] \,=\, 0
\eeq
by the Serre relations for $\mF$. A related feature is the occurrence 
of linear dependencies in the image of $\cJ_\ell$, which occur as 
consequences of both the Jacobi identities and the Serre relations.
The simplest example is for $u_i\in \mT^{(1)}$
\beq
\cJ_3 (u_1\otimes u_2 \wedge u_3) \,+\,
\cJ_3 (u_3\otimes u_1 \wedge u_2) \,+\,
\cJ_3 (u_2\otimes u_3 \wedge u_1) \,=\, 0 \, .
\eeq
However, the majority of elements of the kernel will not take such a
simple form: when inserting the Jacobi identity anywhere in a multi-commutator
we must first rewrite the multi-commutator in the canonical nested form 
\eqref{eq:cJell} of level-one elements which yields a more complicated
expression,  whose vanishing is not immediately obvious by inspection. Similar
remarks apply to linear dependencies due to the Serre relations. For clarity
of notation we will drop the symbol $\otimes$ after the conversion, {\em viz.}
\beq\label{psiell}
\Psi^{(\ell)} \,= \, \cJ_\ell \; \Psi^{\otimes(\ell)} \, \in \, \mF^{(\ell)}
\eeq
We will
generally refer to a tensor state $\Psi^{\otimes(\ell)}$ that maps to zero,
{\em i.e.} for which
\beq\label{virtual}
\cJ_\ell \; \Psi^{\otimes(\ell)} \,=\, 0 \, ,
\eeq
as a {\em virtual state}. The main challenge in understanding 
$\mF^{(\ell)}$ is thus to understand the kernel of $\cJ_\ell$, or
equivalently the characterization of virtual states in $\mTh^{(\ell)}$. 

We shall see below in section \ref{sec:4.4} that all MTGs 
(at every level $> 1$) map 
to zero in $\mF^{(\ell)}$, which also means that the associated 
affine modules are absent. One therefore has to apply Virasoro 
operators to reach affine modules in $\mT^{(\ell)}$ whose 
image in $\mF^{(\ell)}$ does not vanish. The fact that non-vanishing 
elements of (combinations of states in) the Virasoro 
module can be mapped to zero leads to presence of `Virasoro holes' in
the KMA, and it is this feature which destroys the
Virasoro module structure of $\mF^{(\ell)}$.

\subsection{Converting tensor products to multi-commutators}
\label{sec:4.1}

For the conversion of tensor products into KMA elements we employ 
the vertex operator algebra (VOA) prescription, and the fact that
the KMA can be embedded in a Hilbert space $\mH$ of 
physical string states \cite{GN} as explained in \cite{CKMN}:
this is the quotient space
\beq\label{eq:mH}
\mH \,\coloneqq\, \mP_1 / \rL_{-1}\mP_0 \, ,
\eeq
where the spaces $\mP_n$ for $n=0,-1$ are defined by
\beq
\vp\,\in\,\mP_n \quad \Leftrightarrow \quad \rL_m\vp \,=\, 0 
\quad (m\geq 1) \quad \text{and} 
\quad (\rL_0 - n)\vp\,=\,0 \,.
\eeq
Here, we are employing standard notation from string theory~\cite{Scherk} 
with $\rL_m = \frac12 \sum_n :\alpha_{n\mu} \alpha_{m-n}^\mu:$,
where the $\alpha^\mu_m$ are the usual string oscillators (for $\mu=0,1,2$). 
For all levels, the affine generators are physical in the sense 
that~\footnote{Notice that this statement is very different from (\ref{eq:LmT})\,!}
\beq
\big[ \rL_m \,,\, T_n \big] 
\,=\, 0
\eeq
for all $m,n\in \mathbb{Z}$. Hence, all affine, and therefore all Sugawara
Virasoro actions preserve $\mH$.

The commutator between any two elements $\vp, \psi \in \mH$ is 
defined via the state-operator correspondence through the formula
\beq\label{eq:StateOperator}
\cJ_2 ( \vp\otimes \psi) \,=\, [\vp\,,\, \psi ] \,\coloneqq\,
\oint\frac{\mathrm{d}z}{2\pi i} \ \cV(\vp;z) \, \psi \, ,
\eeq
where $\cV(\vp;z)$ is the vertex operator associated to the 
state $\vp\in\mH$. Similarly, multiple commutators appearing in $\cJ_\ell$
correspond to the iterated application of vertex operators.
As shown in~\cite{Borcherds,IF,FLM} (see also \cite{GN})
this definition satisfies all the requisite 
properties of a Lie bracket, to wit, antisymmetry and the Jacobi identity, 
{\em modulo} elements of $\rL_{-1}\mP_0$. Furthermore, it automatically
takes care of the Serre relations, in the sense that \eqref{eq:StateOperator}
will simply give zero on the quotient (\ref{eq:mH}) if the Serre relation 
is anywhere contained in a
multi-commutator. In other words, in the VOA formalism we need not
worry about either Jacobi identities or Serre relations. In particular
using the Free Lie Algebra as an intermediate step in the construction
of $\mF$ as in more conventional approaches is not necessary anymore.
In practice, the evaluation of all multi-commutators by means of 
\eqref{eq:StateOperator} is done with the Mathematica package \cite{DDFPackage}
for all examples presented below.

When descending from $\mTh^{(\ell)}$ to $\mF^{(\ell)}$ by converting
tensor products into Lie algebra commutators, we return from the 
multi-string Fock space to a subspace of the one-string Hilbert space $\mH$, 
which, however, is now much larger than $\mT^{(1)}$. In particular, for 
$\ell>1$ this larger subspace
comprises both transversal {\em and} longitudinal states built on the 
level-$\ell$ tachyonic states with momenta
\beq\label{eq:awl}
\ba_w^{(\ell)} \,=\, -\, \ell\brr - \left(\ell \,+\, \frac{w^2 -1}{\ell} \right) 
\bd \,+\, w \br_1 \, , \qquad w\in\mathbb{Z} \, ,
\eeq 
which in general do not belong to the $\mF$ root lattice any more.
In summary, the space of physical states $\mH$ at level $\ell$ is the 
linear span of state
\beq
\prod_{i=1}^M \Al_{-m_i} \prod_{j=1}^N \Bl_{- n_j} \, \ket{ \ba^{(\ell)}_n} 
\qquad
\text{for} \quad 
\begin{array}{l}
m_1 \ge \ldots \ge m_M \ge 1 \, , \\
n_1 \ \ge \ldots \ge n_N \  \ge 2 \, , 
\end{array} \quad
\quad 
\text{and}
\quad 
M, N \ge 0 \, ,
\eeq
where the transversal and longitudinal DDF operators $\Al_m$ and $\Bl_m$
are written out explicitly in section~3 of \cite{CKMN}. Note that application
of the level-$\ell$ DDF operators shifts the momentum by the fractional
amount $m\bd/\ell$. While the states in $\mT^{(\ell)}$ are built on 
discrete momenta corresponding to the set of level-one roots, the discreteness 
is thereby diluted at higher levels by the need to introduce intermediate 
momenta between root lattice points, which fill the continuum more 
and more densely as $\ell\rightarrow\infty$. Even though these 
intermediate momenta do not appear in $\mF$, for which all momenta
must lie on the root lattice, their presence is indispensable for the present
approach. Let us also mention that there is a multi-string perspective
on the computation of commutators via string scattering \cite{GNW}.

\subsection{MTG descendants and Lie algebra elements}
\label{sec:4.2}

For illustration we map some of the MTGs from subsection \ref{sec:3.2}
to $\mF$.  For each of the MTG we are interested in identifying the the 
corresponding highest weight state in $\mF$. Since all the MTGs are virtual, 
{\em i.e.} obey (\ref{virtual}), we must add some coset Virasoro operator
insertions to obtain non-vanishing states. 
Below, in the subsections
\ref{sec:4.3} - \ref{sec:4.4} we address the question which insertions
yield unique non-vanishing elements in $\mF$ (after applying $\cJ_\ell$).
In this subsection we will simply list some results.

It is useful to combine all $\ell-1$ coset Virasoro operators acting on the
states of $\mT^{(\ell)}$ into on {\em coset Virasoro tower operator}. To this
end let $\{\bm_1 , \ldots ,\bm_{\ell-1}\}$ be a
set of multi-indices with 
$\bm_i = (m_{i,1}, \ldots, m_{i,n_i}) \in \mathbb{Z}^{n_i}$
and we define the operator
\beq
\cLt{\ell}{\{\bm_1, \ldots, \bm_{\ell-1}\}} \psi^{\otimes (\ell)} \,=\, 
\cL{\ell}{\bm_1} 
\Big( \bE \otimes \cL{\ell-1}{\bm_2} \Big( \ldots
\Big( \bE \otimes \ldots \otimes \bE \otimes 
\cL{2}{\bm_{\ell-1}}  \psi^{\otimes (\ell)}  \Big) \ldots \Big)\Big)\, , 
\eeq
where each multi-indexed coset Virasoro operator is given 
by the consecutive actions
\beq
\underbrace{\bE \otimes \ldots \otimes \bE}_{(\ell-k) \text{ factors}} \otimes \cL{k}{\bm_{\ell-k+1}} \psi^{\otimes (\ell)}  
\,=\, 
\bE \otimes \ldots \otimes \bE \otimes 
\cL{k}{m_{\ell-k+1,1}} \cdots \cL{k}{m_{\ell-k+1,n_{\ell-k+1}}} 
\psi^{\otimes (\ell)} \, . 
\eeq
For multi-indices of length 1 in $\{\bm_1 , \ldots ,\bm_{\ell-1}\}$ we simply
write normal indices, {\em i.e.} without braces. To label the
absence of level-$k$ coset Virasoro operators in the coset Virasoro tower
operator we introduce the special notation
$\{\bm_1, \ldots , \bm_{\ell-k} , \bullet, \bm_{\ell-k+2}, \ldots , \bm_{\ell-1}\}$
with
\beq
 \bE \otimes \ldots \otimes \bE \otimes \cL{k}{\bullet}  \psi^{\otimes (\ell)}  \,=\, 
\psi^{\otimes (\ell)} \, . 
\eeq
We will see below that there are linear relations between different
coset Virasoro tower operators related to the null vectors of the associated
Virasoro representations. The affine highest weight states of $\mF^{\ell}$ 
are related to the MTGs of $\mTh^{(\ell)}$ via the operator actions
\begin{align}
\cJ_\ell \ \Big( \cLt{\ell}{\{\bm_1, \ldots, \bm_{\ell-1}\}} \Psi^{\otimes (\ell)} \Big)
\end{align}
with all possible sets of multi-indices $\{\bm_1, \ldots, \bm_{\ell-1}\}$.

The descendant of the level 2 MTG in~\eqref{eq:mtg2} is related to the highest
weight state of $\mF^{(2)}$ via a coset Virasoro insertion of degree one, 
{\em i.e.}
\beq
\cJ_2 \, \Big( \cLt{2}{\{-1\}} \, \Psi_{1,\bD,\pm1}^{\otimes(2)} \Big) 
\,=\, \cJ_2 \, \Big( \cL{2}{-1} \, \Psi_{1,\bD,\pm1}^{\otimes(2)} \Big)
\,=\, 4 \ket{\ba_{\pm 1}^{(2)}} \, .
\eeq
Similarly the two level 3 MTGs \eqref{eq:mtg31} and \eqref{eq:mtg32}
are related to the highest weight singlet respectively triplet of $\mF^{(3)}$ 
via a coset Virasoro insertion of degree 2, {\em i.e.}
\begin{align}
\begin{aligned}
\cJ_3 \, \Big( \cLt{3}{\{-1,-1\}}  
\, \Psi_{2,\bE,0}^{\otimes(3)} \Big) 
&=\, \cJ_3 \, \Big( \cL{3}{-1} \, \left( \bE \otimes \cL{2}{-1} \right) 
 \, \Psi_{2,\bE,0}^{\otimes(3)} \Big) \\ 
&=\, \Big[-\, \frac{8}{3} \ \Ad_{-2} \Ad_{-2} 
\,+\, \ 6 \Ad_{-1} \Ad_{-1} \Ad_{-1} \Ad_{-1} \\
&\quad \ \ -\, \frac{28}{3} \ \Ad_{-1} \Ad_{-1} \Bd_{-2} \\
&\quad \ \ +\, \frac{7}{9} \ \Bd_{-2} \Bd_{-2} \,-\, \frac{7}{18} \ \Bd_{-4} \Big] 
\ket{\ba_0^{(3)}} 
\end{aligned}
\end{align}
and
\beq
\cJ_3 \, \Big( \cLt{3}{\{-1,-1\}}  \, \Psi_{1,\bD,1}^{\otimes(3)} \Big)
 \,=\, -\, 4 \ket{\ba_1^{(3)}} \, .
\eeq
The highest weight triplet MTG on level 4 \eqref{eq:mtg44} is related to the
highest weight state in $\mF^{(4)}$ via
\beq
\cJ_4 \,  \Big( \cLt{4}{\{-1,-1,-1\}}    \, \Psi_{1,\bD,1}^{\otimes(4)} \Big)  \,=\, 
\frac{16}{5} \ket{\ba_1^{(4)}} \, .
\eeq
For the two singlet MTGs \eqref{eq:mtg41} and \eqref{eq:mtg43} we find
\beq\label{eq:singlet0}
\cJ_4 \,  \Big( \cLt{4}{\{-1,-1,-1\}} 
 \, \Psi_{2,\bE,0}^{\otimes(4)} \Big) \,=\, 0
\eeq
and
\begin{align}
\begin{aligned}
\cJ_4 \,  &\Big(\cLt{4}{\{-1,-1,-1\}}  \, \bar \Psi_{2,\bE,1}^{\otimes(4)} \Big) \\
 \,=\, \Bigg[ &-\,\frac{63}{20} \ \Bv_{-5} \,-\, 21 \ \Av_{-3} \Av_{-2} 
\,+\, \frac{63}{20} \ \Bv_{-3} \Bv_{-2} 
\,-\, \frac{126}{5} \ \Av_{-2} \Av_{-1} \Bv_{-2} \\
&\,-\,\frac{126}{5} \ \Av_{-1} \Av_{-1} \Bv_{-3} 
\,+\, \frac{224}{5} \ \Av_{-2} \Av_{-1} \Av_{-1} \Av_{-1} \Bigg] 
\ket{\ba_{0}^{(4)}} \, .
\end{aligned}
\end{align}
For the first fiveplet MTG \eqref{eq:mtg43} we find
\beq\label{eq:fiveplet0}
\cJ_4 \,  \Big( \cLt{4}{\{-1,-1,-1\}}  \, \Psi_{2,\bF,2}^{\otimes(4)}  \Big) \,=\, 0 \, .
\eeq
For the other triplet MTG on level 4 \eqref{eq:mtg42} we find
\begin{align}
\begin{aligned}
\cJ_4 \,  \Big( \cLt{4}{\{\bullet,-1,-1\}}  \,  \Psi_{3,\bD,1}^{\otimes(4)} \Big)  
&\,=\, \cJ_4 \, \left( \bE \otimes \cL{3}{-1} \, 
\left( \bE \otimes \bE \otimes \cL{2}{-1} 
 \, \Psi_{3,\bD,1}^{\otimes(4)} \right) \right)\\
&\,=\, \Bigg[-\, \frac{45\sqrt{2}}{16} \ \Bv_{-4} 
\,-\, \frac{576\sqrt{2}}{16} \ \Av_{-3} \Av_{-1} 
\,-\, \frac{201\sqrt{2}}{16} \ \Av_{-2} \Av_{-2} \\
&\quad \quad \,-\, \frac{783}{16} \ \Av_{-2} \Bv_{-2} 
\,-\, \frac{45}{4} \ \Av_{-1} \Bv_{-3} 
\,+\, \frac{975\sqrt{2}}{128} \ \Bv_{-2} \Bv_{-2}  \\
&\quad \quad \,+\, \frac{717}{8} \ \Av_{-2} \Av_{-1} \Av_{-1}
\,-\, \frac{771\sqrt{2}}{32} \ \Av_{-1} \Av_{-1} \Bv_{-2} \\
&\quad \quad \,-\, 
\frac{915\sqrt{2}}{32} \ \Av_{-1} \Av_{-1} \Av_{-1} \Av_{-1} \Bigg]
\ket{\ba_{1}^{(4)}} \, .
\end{aligned}
\end{align}
While the the insertions on level 2 and 3 follow a clear pattern, on level 4 
this is no longer the case. We will give a partial explanation for this behavior 
in subsection \ref{sec:4.4}. All these results are in perfect agreement with the
results presented in \cite{CKMN} (up to unimportant normalizations).

The level 5 states of $\mF$ are currently out of reach for our computational
tools.

\subsection{Back to the Feingold-Frenkel algebra}
\label{sec:4.3}

We have already pointed out that the KMA $\mF$ can be obtained from the
tensor algebra $\mT$ by first computing the Free Lie Algebra $F$ and 
subsequently dividing out the ideal $\mJ$ generated by the Serre 
relations. Here we briefly review these steps and collect some results
from the literature.

In the following let $F = \bigoplus_{\ell \in \mathbb{N}} F^{(\ell)}$ be the
Free Lie Algebra and $\rI_\ell \colon \mT^{(\ell)} \to F^{(\ell)}$ the map from
the tensor algebra to the Free Lie Algebra.\footnote{Notice that $\rI_\ell$ is
different from the map $\cI_\ell$ introduced above that maps directly to the
KMA $\mF^{(\ell)}$.} The kernel of $\rI_\ell$ consists of those states related
by anti-symmetry and the Jacobi identity. In \cite{BB} it is explained how to 
determine the Free Lie Algebra $F$ to any level $\ell$. Formally, we can write
\beq\label{eq:FreeLA}
F^{(\ell+1)} \,=\, L \otimes F^{(\ell)} \,-\, \mathrm{Ker} \; \rI_{\ell + 1} 
\eeq
with $F^{(1)} = L$. Here and in the following we use a minus sign to indicate
the quotient. To reach the KMA $\mF$ we must subsequently divide out the
Serre ideal
\beq
\mJ \,=\, \mJ_+ \oplus \mJ_-, \quad \text{with} \quad
\mJ_\pm \,=\,  \bigoplus_{\ell\geq 2} \mJ_{\pm \ell}
\eeq
generated by $\mJ_2 = L(2\bL_1 + 3 \bd)$. This step is highly non-trivial
because this ideal has a non-zero intersection with the kernel of $\rI_\ell$.
Finding the intersection requires an extremely elaborate analysis and so far
has only been achieved up to level 5 \cite{Kang3}. Formally, it is easy
enough to write out the analog of \eqref{eq:FreeLA} for the KMA \cite{BB}
\beq
\mathfrak{F}^{(\ell+1)} \,=\, L \otimes F^{(\ell)} \,-\, \mathrm{Ker} 
\; \rI_{\ell + 1} \,-\, \mJ_{\ell + 1}
\eeq
with
\beq
\mJ_{\ell + 1} \,=\, L \otimes \mJ_\ell 
\,-\, (L \otimes \mJ_\ell) \cap \mathrm{Ker} \; \rI_{\ell + 1} \, .
\eeq
For the first five levels of $\mF$ we find the vector space isomorphisms
\begin{align}
\begin{aligned}\label{eq:FLevels}
\mF^{(1)} &\,=\, \mTh^{(1)} \, , \\
\mF^{(2)} &\,\cong \, \mTh^{(2)} \,-\, \mJ_2 \, , \\
\mF^{(3)} &\, \cong \,\mTh^{(3)} \,-\, \wedge^3 L \,-\, \mJ_3 \, , \\
\mF^{(4)} &\, \cong \, \mTh^{(4)} \,-\, L \otimes (\wedge^3 L) 
\,-\, (S^2(\wedge^2 L) 
\,-\, \wedge^4 L) \,-\, \mJ_4 \, , \\
\mF^{(5)} &\, \cong \, \mTh^{(5)} \,-\, L \otimes L \otimes (\wedge^3 L) 
\,-\, (L \otimes S^2(\wedge^2 L) \,-\, L \otimes (\wedge^4 L)) 
\,-\, L \otimes (\wedge^2 L) \wedge (\wedge^2 L) \\
&\quad \,-\, (\wedge^3 L) \otimes (\wedge^2 L) \,-\, \wedge^5 L \,-\, \mJ_5 \, .
\end{aligned}
\end{align}
The loss of the Virasoro representation structure of the l.h.s. here is
manifest from the fact that the terms subtracted on the r.h.s. of 
\eqref{eq:FLevels} do not constitute Virasoro modules. The low level ideals are
given by
\begin{align}
\begin{aligned}
\mJ_3 &\,=\, L \otimes \mJ_2 \,-\, (L \otimes \mJ_2) \cap (\wedge^3 L) \, , \\
\mJ_4 &\,=\, L \otimes \mJ_3 \,-\, (L \otimes \mJ_3) \cap 
\left( S^2(\wedge^2 L) \,-\, \wedge^4 L \right) \, , \\
\mJ_5 &\,=\, L \otimes \mJ_4 \,-\, (L \otimes \mJ_4) \cap 
\left( L \otimes (\wedge^2 L) \wedge (\wedge^2 L) 
\oplus (\wedge^3 L) \otimes (\wedge^2 L) \oplus \wedge^5 L \right) \, .
\end{aligned}
\end{align}
While these formulas can be obtained systematically and in closed form to any
desired level, the main difficulty here is in actually computing the
intersections. For $\mJ_3$ the intersection term can be obtained from a general
anlysis of the two vector spaces involved and the Virasoro representations
acting on them (see \cite{BB} for details). For higher $\mJ_\ell$, 
Kang \cite{Kang3} has developed a homological theory and used Hochschild-Serre
spectral sequences to determine the intersections. For $\mF$ he presented
results up to level 5
\begin{align}
\begin{aligned}
&(L \otimes \mJ_2) \cap (\wedge^3 L ) \,=\, 0 \, , \\
&(L \otimes \mJ_3) \cap \left( S^2(\wedge^2 L) \,-\, \wedge^4 L \right) 
\,=\, L(4\bL_1 + 5\bd) \oplus S^2(J_2) \, ,\\
&(L \otimes \mJ_4) \cap \left( L \otimes (\wedge^2 L) \wedge (\wedge^2 L) 
\oplus (\wedge^3 L) \otimes (\wedge^2 L) \oplus \wedge^5 L \right) 
\,=\, L \otimes (\wedge^2 \mJ_2) \, .
\end{aligned}
\end{align}
Our construction of using the DDF states and a vertex operator algebra has the
advantage, that we have an automatic map $\cJ_\ell$ from $\mTh^{(\ell)}$ to
$\mF^{(\ell)}$ at any level $\ell$. However, since $\mTh^{(\ell)}$ is much
larger than $\mF^{(\ell)}$ many states are either in the kernel of $\cJ_\ell$ or
are mapped to the same elements in $\mF^{(\ell)}$. Hence, combining the DDF 
construction with the traditional approach 
\eqref{eq:F} can provide a way to obtain a minimal set of states in
$\mTh^{(\ell)}$ on which the map $\cJ_\ell$ is bijective.

We note that the apparent simplicity of formulas like \eqref{eq:FLevels}
is a bit deceptive, as they do not contain enough information to perform
the formal subtraction of representations in case representations
appear several times (although they may suffice for the 
computation of multiplicities). This difficulty was already noted in \cite{BB}.
It is here that the DDF representation gives a much better handle on this
problem, because there are no such ambiguities in the DDF states.

In the following, we give a rough outline of the program to determine the
minimal set of states in $\mTh$ on which $\cJ_\ell$ is bijective. In particular
we will provide some explanations of the vanishing of
\eqref{eq:singlet0} and \eqref{eq:fiveplet0}. However, many questions must
remain unanswered at the moment and we leave the further elaboration of 
the results in the following subsections for future work.

\subsection{Some examples}
\label{sec:4.4}

For $|\ell|\geq 2$, each $\mF^{(\ell)}$ consists of an infinite direct sum of
affine modules. The action of the affine generators $T_m$ commutes with the 
map $\cJ_\ell$, as follows immediately from the distributivity of 
the affine action on tensor products and commutators. Therefore, 
it is enough to map the affine ground state of each of these 
affine modules into $\mF$.

To simplify the following discussion, we will look at a concrete example,
namely, the triplet modules on level 3. In the tensor algebra these modules are
described by
\beq\label{eq:T3Term}
\Vir(\tfrac{7}{10},\tfrac{1}{10}) 
\otimes \Vir(\tfrac{1}{2},\tfrac{1}{2}) 
\otimes L(\bL_0 + 2 \bL_1 + 5\bd) \, .
\eeq
The MTG of these modules is \eqref{eq:mtg31}
\beq
\Psi_{1,\mathbf{3},1}^{\otimes(3)} 
\,=\, \ket{\ba_0^1} \otimes \ket{\ba_1^1} \wedge \ket{\ba_0^1}
\eeq
and we obtain the set of all affine triplet ground states of $\mT^{(3)}$ by the
action of $\cLt{3}{\bm_1,\bm_2}$ on 
$\Psi_{1,\mathbf{3},1}^{\otimes(3)}$ for all $\bm_1, \bm_2$. 

The map $\cJ_3$ has a non-trivial kernel on \eqref{eq:T3Term}. In particular,
it is easy to see that $\cJ_3 \Psi_{1,\bD,1}^{\otimes(3)} = 0$ as the
ground state momentum to this state obeys $(-3\brr - \bd + \br_1)^2 >2$.
In the following we will explain how \eqref{eq:F} helps us to 
determine this kernel without having
to do any actual calculations. 

The first observation is that there is a one-to-one correspondence between
Virasoro character and applications of coset Virasoro operators $\cL{\ell}{-m}$.
For example consider the Virasoro character (see \eqref{eq:VirCh})
\beq
\Ch \Vir(\tfrac{7}{10},\tfrac{1}{10}) 
\,=\, 1 \,+\, q \,+\, q^2 \,+\, q^3 \,+\, 2 q^4 \,+\, \ldots 
\eeq
Each term $n q^m$ tells us that there are $n$ linearly independent combinations
of Virasoro operators $\cL{3}{-m}$, $\cL{3}{-m_1} \cL{3}{-m_2}, \ldots$ with
total degree $m$. For example the coefficient of $q^2$ is 1, so $\cL{3}{-2}$
and $\cL{3}{-1} \ \cL{3}{-1}$ must be related. Or in other words the Virasoro
representation $\Vir(\tfrac{7}{10},\tfrac{1}{10})$ contains a null vector at
degree 2. Indeed we find
\beq
\left( \cL{3}{-2} \,-\, \frac{5}{4} \ \cL{3}{-1} \ \cL{3}{-1} \right) 
\Psi_{1,\mathbf{3},1}^{\otimes(3)} \,=\, 0 \, ,
\eeq
which just reproduces the usual null vector of the Virasoro module.
Similarly, we find that also the level 2 Virasoro representation 
$\Vir(\tfrac{1}{2},\tfrac{1}{2})$ has a null vector at degree 2.
Multiplying the two characters yields
\begin{align}
\begin{aligned}\label{eq:qSeries1}
&\Ch \Vir(\tfrac{7}{10},\tfrac{1}{10}) \
\Ch \Vir(\tfrac{1}{2},\tfrac{1}{2}) \\
&\,=\, (1 \,+\, q \,+\, q^2 \,+\, q^3 \,+\, 2 q^4 \,+\, \ldots ) \ 
(1 \,+\, q \,+\, q^2 \,+\, q^3 \,+\, 2 q^4 \,+\, \ldots ) \\
&\,=\, 1 \,+\, 2 q \,+\, 3q^2 \,+\, 5 q^3 \,+\, 9 q^4 \,+\, \ldots 
\end{aligned}
\end{align}
Thus, we deduce that the following three states must be
linearly independent in $\mT^{(3)}$
\beq\label{eq:3States}
\cLt{3}{\{-2,\bullet\}} \Psi_{1,\mathbf{3},1}^{\otimes(3)} \, , \quad 
 \cLt{3}{\{\bullet,-2\}}  \Psi_{1,\mathbf{3},1}^{\otimes(3)} \, , \quad 
\cLt{3}{\{-1,-1\}} \Psi_{1,\mathbf{3},1}^{\otimes(3)} \, .
\eeq
An explicit calculation confirms this. To answer the question
which of these states belong to the minimal set of states from
$\mTh^{(3)}$, on which the action of $\mJ_3$ is bijective, we
study \eqref{eq:F}. 

Reducing $\mTh^{(\ell)}$ to the Free Lie Algebra $F^{(\ell)}$ and subsequently
the Kac-Moody algebra $\mF^{(\ell)}$ only affects the Virasoro algebras
that multiply the affine modules but not the modules themselves.
Concretely the characters of the the Virasoro algebras
receive subtractions due to the Jacobi identity and the Serre relation.
In our example \eqref{eq:qSeries1} then becomes
\begin{align}
\begin{aligned}\label{eq:qSeries2}
&\Ch \Vir(\tfrac{7}{10},\tfrac{1}{10}) \,
\left( \Ch \Vir(\tfrac{1}{2},\tfrac{1}{2}) 
{\color{myblue} -1 } \right) 
 \color{myred} - (q \,+\, q^2 \,+\, 2q^3 \,+\, 3q^4 \,+\, \ldots)\\
&\,=\, (1 \,+\, q \,+\, q^2 \,+\, q^3 \,+\, 2 q^4 \,+\, \ldots ) \ 
 (q \,+\, q^2 \,+\, q^3 \,+\, 2 q^4 \,+\, \ldots ) 
 \color{myred} - (q \,+\, q^2 \,+\, 2q^3 \,+\, 3q^4 \,+\, \ldots) \\
&\,=\, \, q^2 \,+\, q^3 \,+\, q^4 \,+\, \ldots \, .
\end{aligned}
\end{align}
Here the {\color{myblue}blue} term is due to the Serre ideal 
$\mJ_2 = L \otimes \mJ_2$ and the {\color{myred}red} term is due to Jacobi
identity $\wedge^3L$. The $q$-series are easily obtained by evaluating
the characters of $L \otimes \mJ_2$ and $\wedge^3L$. In general such terms 
do not have a nice representation in terms of the character of 
minimal model Virasoro algebras.

Several observations can be made from \eqref{eq:qSeries2}. Firstly, the states
\beq
\Psi_{1,\mathbf{3},1}^{\otimes(3)} \, , \quad
\cLt{3}{\{-1,\bullet\}} \Psi_{1,\mathbf{3},1}^{\otimes(3)} \, , \quad 
\cLt{3}{\{\bullet,-1\}} \Psi_{1,\mathbf{3},1}^{\otimes(3)} 
\eeq
are all virtual, {\em i.e.} they are in the kernel of $\cJ_3$ because
the coefficients of $q^0$ and $q^1$ in \eqref{eq:qSeries2} are 0.
Since the affine generators commute with $\cJ_\ell$, the entire
affine modules associated to these three states are virtual.

From the {\color{myblue} $-1$} subtraction of the level 2 Virasoro
character we learn that all states which do not contain at
least one $\cL{2}{-m}$ with $m\ge 1$ are also in the kernel of 
$\cJ_3$. Thus, without any calculation we now know that 
\beq
\cJ_3 \left(\cLt{3}{\{-2,\bullet\}} \Psi_{1,\mathbf{3},1}^{\otimes(3)} \right) 
\,=\, 0 \, .
\eeq
This observation extends to all levels and we can conclude
that for this reason all MTGs must be virtual, {\em i.e.} in the
kernel of $\cJ_\ell$.

Finally, the coefficient of $q^2$ in \eqref{eq:qSeries2} is 1, so when mapped
to $\mF^{(3)}$ the other two states in \eqref{eq:3States} will be related.
Indeed an explicit calculations shows that 
\beq
\cJ_3 \left(  \cLt{3}{\{\bullet,-2\}} \Psi_{1,\mathbf{3},1}^{\otimes(3)} \right) 
\,=\, 3 \ \cJ_3 \left( \cLt{3}{\{-1,-1\}} \Psi_{1,\mathbf{3},1}^{\otimes(3)} \right)
\,=\, -\, 12 \ket{\ba_1^{(3)}} \, .
\eeq
Hence we conclude that of the three states in \eqref{eq:3States} only the last one
is relevant for the construction of the KMA $\mF$. 

Repeating the above analysis on level 4 reveals that these arguments
become a lot more cumbersome on higher levels. 
Using \eqref{eq:F} as well as our explicit results for the characters in
\cite{CKMN} we obtain the following expression for the character of $\mF^{(4)}$

\begin{footnotesize}
\begin{align}\label{eq:Ch4Long}
\begin{aligned}
\Ch \mF^{(4)} &\,=\, \bigg[ \Ch \Vir(\tfrac{4}{5},0) \\
&\quad \quad \quad \times \bigg( \Ch \Vir(\tfrac{7}{10},\tfrac{3}{2}) 
\left( \Ch \Vir(\tfrac{1}{2},\tfrac{1}{2}) {\color{myblue} \,-\, 1} \right) 
{\color{myred}\,-\, \left( q \,+\, q^2 \,+\, 2 q^3 \,+\, 3 q^4 \,+\, 5 q^5 
\,+\, 7 q^6 \,+\, 12 q^7 \,+\, 16 q^8 \,+\, \ldots \right)} \bigg) \\
&\quad \quad \,+\, \Ch \Vir(\tfrac{4}{5},\tfrac{7}{5}) \\
&\quad \quad \quad \times 
\bigg( \Ch \Vir(\tfrac{7}{10},\tfrac{1}{10}) 
\left( \Ch \Vir(\tfrac{1}{2},\tfrac{1}{2}) {\color{myblue} \,-\, 1} \right)
{\color{myred}\,-\, \left( q \,+\, q^2 \,+\, 2 q^3 \,+\, 3 q^4 \,+\, 5 q^5 
\,+\, 7 q^6 \,+\, 12 q^7 \,+\, 17 q^8 \,+\, \ldots \right)} \bigg) \\
&\quad \quad {\color{mygreen} \,-\, \left( 2 q^2 \,+\, 3 q^3 \,+\, 7 q^4 
\,+\, 12 q^5 \,+\, 25 q^6 \,+\, 39 q^7 \,+\, 71 q^8 \,+\, \ldots \right)} \bigg] 
\Ch L(4\bL_0 + 6\bd) \\
&\quad \,+\, \bigg[ \Ch \Vir(\tfrac{4}{5},\tfrac{2}{3}) \\
&\quad \quad \quad \times \bigg( \Ch \Vir(\tfrac{7}{10},\tfrac{3}{2}) 
\left( \Ch \Vir(\tfrac{1}{2},\tfrac{1}{2}) {\color{myblue} \,-\, 1} \right)
{\color{myred}\,-\, \left( q \,+\, q^2 \,+\, 2 q^3 \,+\, 3 q^4 \,+\, 5 q^5 
\,+\, 7 q^6 \,+\, 12 q^7 \,+\, 16 q^8 \,+\, \ldots \right)} \bigg) \\
&\quad \quad \,+\, q^2 \, \Ch \Vir(\tfrac{4}{5},\tfrac{1}{15}) \\
&\quad \quad \quad \times \bigg( \Ch \Vir(\tfrac{7}{10},\tfrac{1}{10}) 
\left( \Ch \Vir(\tfrac{1}{2},\tfrac{1}{2}) {\color{myblue} \,-\, 1} \right)
{\color{myred}\,-\, \left( q \,+\, q^2 \,+\, 2 q^3 \,+\, 3 q^4 \,+\, 5 q^5 
\,+\, 7 q^6 \,+\, 12 q^7 \,+\, 17 q^8 \,+\, \ldots \right)} \bigg) \\
&\quad \quad {\color{mygreen}\,-\, \left( q^2 \,+\, q^3 \,+\, 4 q^4 
\,+\, 8 q^5 \,+\, 17 q^6 \,+\, 31 q^7 \,+\, 58 q^8 \,+\, \ldots \right) } \bigg]
\Ch L(2\bL_0 + 2 \bL_1 + 7\bd) \\
&\quad \,+\, \bigg[ \Ch \Vir(\tfrac{4}{5},3) \\
&\quad \quad \quad \times \bigg( \Ch \Vir(\tfrac{7}{10},\tfrac{3}{2}) 
\left( \Ch \Vir(\tfrac{1}{2},\tfrac{1}{2}) {\color{myblue} \,-\, 1} \right)
{\color{myred}\,-\, \left( q \,+\, q^2 \,+\, 2 q^3 \,+\, 3 q^4 \,+\, 5 q^5 
\,+\, 7 q^6 \,+\, 12 q^7 \,+\, 16 q^8 \,+\, \ldots \right)} \bigg) \\
&\quad \quad \,+\, q^4 \, \Ch \Vir(\tfrac{4}{5},\tfrac{2}{5}) \\
&\quad \quad \quad \times \bigg( \Ch \Vir(\tfrac{7}{10},\tfrac{1}{10}) 
\left( \Ch \Vir(\tfrac{1}{2},\tfrac{1}{2}) {\color{myblue} \,-\, 1} \right)
{\color{myred} \,-\, \left( q \,+\, q^2 \,+\, 2 q^3 \,+\, 3 q^4 \,+\, 5 q^5 
\,+\, 7 q^6 \,+\, 12 q^7 \,+\, 17 q^8 \,+\, \ldots \right)} \bigg) \\
&\quad \quad {\color{mygreen}\,-\, \left( q^2 \,+\, 2 q^3 \,+\, 4 q^4 
\,+\, 7 q^5 \,+\, 14 q^6 \,+\, 24 q^7 \,+\, 42 q^8 \,+\, \ldots \right) } \bigg]
 \Ch L(4\bL_1 + 6\bd) \, .
\end{aligned}
\end{align}
\end{footnotesize}\noindent
On level 4 each of the three affine modules appears with two different towers of
coset Virasoro algebras. As we have seen explicitly in section \ref{sec:3.2}
there are thus six MTGs. The subtractions associated to
\beq
{\color{myred} L \otimes (\wedge^3 L) } 
\oplus {\color{myblue} L \otimes L \otimes \mJ_2}
\eeq
can be clearly associated with these towers. These are the {\color{myblue}blue}
and {\color{myred}red} terms in \eqref{eq:Ch4Long}. The {\color{mygreen}green}
subtractions, however, that come from
\beq
{\color{mygreen} S^2(\wedge^2 L) \,-\, \wedge^4 L \oplus (L \otimes \mJ_3) \cap 
\left( S^2(\wedge^2 L) \,-\, \wedge^4 L \right)  
\,=\, S^2(\wedge^2 L) \oplus S^2(\mJ_2) \oplus L(4\bL_1 + 2 \bd) 
\,-\, \wedge^4 L }
\eeq
cannot be allocated to the different Virasoro towers but only the different
affine modules. Expanding all the coset Virasoro characters in 
\eqref{eq:Ch4Long} yields
\begin{align}
\begin{aligned}\label{eq:Ch4}
\Ch \mF^{(4)} &\,=\, \left( q^3 \,+\, 4 q^4 \,+\, 9 q^5 \,+\, 20 q^6 
\,+\, 41 q^7 \,+\, 78 q^8 \,+\, \ldots \right) 
\Ch L(4\bL_0 + 6\bd) \\
&\quad \,+\, \left( q^3 \,+\, 3 q^4 \,+\, 8 q^5 \,+\, 19 q^6 \,+\, 39 q^7 
\,+\, 77 q^8 \,+\, \ldots \right) 
\Ch L(2\bL_0 + 2 \bL_1 + 7 \bd) \\
&\quad \,+\, \left( q^4 \,+\, 4 q^5 \,+\, 9 q^6 \,+\, 20 q^7 \,+\, 41 q^8 
\,+\, 78 q^9 \,+\, \ldots \right) 
\Ch L(4\bL_1 + 6\bd) \, .
\end{aligned}
\end{align}
The first line describes the level 4 singlet MTGs 
$\color{myred}\Psi_{2,\bE,0}^{\otimes (4)}$ 
and $\color{myblue}\bar\Psi_{2,\bE,0}^{\otimes (4)}$ (notice the bar on the 
{\color{myblue}blue} MTG). Since the first term in the respective $q$-series is
$q^3$ we obtain the associated highest weight singlet in $\mF^{(4)}$ by acting
with coset Virasoro tower operators of degree 3 on 
$\Psi_{2,\bE,0}^{\otimes (4)}$ and  $\bar\Psi_{2,\bE,0}^{\otimes (4)}$.
However, the coefficient of $q^3$ in the first line of
\eqref{eq:Ch4} is 1. Thus either 
\beq
\cJ_4 \, \left( \cLt{4}{\{-1,-1,-1\}} \, 
{\color{myred} \Psi_{2,\bE,0}^{\otimes(4)} } \right) 
\qquad \text{or} \qquad 
\cJ_4 \,  \left( \cLt{4}{\{-1,-1,-1\}} \,  
{\color{myblue}\bar\Psi_{2,\bE,0}^{\otimes(4)}}  \right) 
\eeq
has to be equal to zero. But without doing the calculation there is no way to
tell which one it is.

Similarly, it can be explained why the triplet MTG \eqref{eq:mtg42} gives 
non-zero  contribution to $\mF^{(4)}$ already for coset Virasoro insertions of
total degree 2 and why the fiveplet MTG \eqref{eq:mtg45} needs coset Virasoro
insertions of total degree 4.

One goal of our future research is to extend and formalize this kind of analysis
and develop it into an algorithm that turns \eqref{eq:F} into a minimal set of
affine tensor ground states that are bijectively mapped to $\mF$ by $\cJ_\ell$.
For the KMA $\mF$ the limiting factor in this approach is the concrete knowledge
of the Serre ideal $\mJ$. The Free Lie Algebra on the other hand can be
constructed to any level. Furthermore, we must understand how to associate the
subtractions with the coset Virasoro towers for levels $\ge 4$. 

Ultimately, the aim is to understand the observations above directly in the DDF
langauge without having to rely on the tradtitional results \eqref{eq:F}.
This will hopefully enable us to get an explicit description of $\mF$ to all
levels.

\appendix
\section{Commuting higher levels}\label{app:example}
In this appendix we provide a concrete example
that $[\mF^{(k)} , \mF^{(\ell -k)}]$ in general is a {\em proper}
subspace of $\mF^{(\ell)}$ for $k > 1$. 
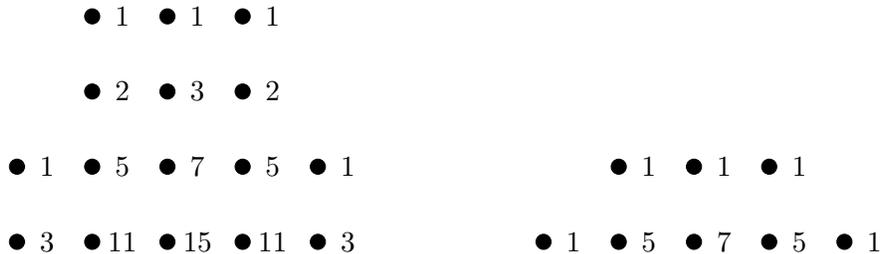
\begin{figure}[H]
\vspace{.5cm}
\begin{tikzpicture}
\draw [fill] (-5,0) circle (.1);
\draw [fill] (-4,0) circle (.1);
\draw [fill] (-3,0) circle (.1);
\draw [fill] (-5,-1) circle (.1);
\draw [fill] (-4,-1) circle (.1);
\draw [fill] (-3,-1) circle (.1);
\draw [fill] (-6,-2) circle (.1);
\draw [fill] (-5,-2) circle (.1);
\draw [fill] (-4,-2) circle (.1);
\draw [fill] (-3,-2) circle (.1);
\draw [fill] (-2,-2) circle (.1);
\draw [fill] (-6,-3) circle (.1);
\draw [fill] (-5,-3) circle (.1);
\draw [fill] (-4,-3) circle (.1);
\draw [fill] (-3,-3) circle (.1);
\draw [fill] (-2,-3) circle (.1);
\node at (-4.6,0) {1};
\node at (-3.6,0) {1};
\node at (-2.6,0) {1};
\node at (-4.6,-1) {2};
\node at (-3.6,-1) {3};
\node at (-2.6,-1) {2};
\node at (-5.6,-2) {1};
\node at (-4.6,-2) {5};
\node at (-3.6,-2) {7};
\node at (-2.6,-2) {5};
\node at (-1.6,-2) {1};
\node at (-5.6,-3) {3};
\node at (-4.6,-3) {11};
\node at (-3.6,-3) {15};
\node at (-2.6,-3) {11};
\node at (-1.6,-3) {3};
%
%
\draw [fill] (2,-2) circle (.1);
\draw [fill] (3,-2) circle (.1);
\draw [fill] (4,-2) circle (.1);
\draw [fill] (1,-3) circle (.1);
\draw [fill] (2,-3) circle (.1);
\draw [fill] (3,-3) circle (.1);
\draw [fill] (4,-3) circle (.1);
\draw [fill] (5,-3) circle (.1);
\node at (2.4,-2) {1};
\node at (3.4,-2) {1};
\node at (4.4,-2) {1};
\node at (1.4,-3) {1};
\node at (2.4,-3) {5};
\node at (3.4,-3) {7};
\node at (4.4,-3) {5};
\node at (5.4,-3) {1};
\end{tikzpicture}
\vspace{.5cm}
\caption{Partial visualization of $\mF^{(2)}$ (left) and $\mF^{(4)}$ (right). 
The top right root of the left character is $[-2,-2,-1]$ and the top root of the
right character is $[-4,-4,-3]$. The numbers in the figure indicate the 
multiplicities.}
\end{figure}\noindent
Recall that $\mF$ has a 
basis in terms of standard multi-commutators, {\em i.e.}
\beq
\mF^{(\ell)} \,=\, \mathrm{span} \{ f_{i_1 \ldots i_n} \ | \ \text{with } 
\ell \text{ generators} \ f_{-1} \ \text{and} \ n \ge \ell \} \, ,
\eeq
which implies that the inclusion 
$[\mF^{(k)}, \mF^{(\ell-k)}] \,\subseteq \, \mF^{(\ell)}$ (for all $k$) is
obvious. To illustrate that for $k>1$
$[\mF^{(k)}, \mF^{(\ell-k)}] \,\subseteq \, \mF^{(\ell)}$ is in general
smaller than $ \mF^{(\ell)}$ we employ the DDF language, as the use of multiple
commutators would entail rather cumbersome expressions. To this aim
we consider the root space of $-[4,5,4]$ to obtain all DDF states 
that can be obtained by commuting level 2 states. With the help of 
\cite{DDFPackage} we find that the level 2 states up to depth 3 are
\begin{align}
\begin{aligned}
-[2,2,1]\colon \quad &&\varphi_{2,\bD,1}^{(2)} 
&\,=\, \ket{\ba_1^{(2)}} \, , \\
-[2,2,2]\colon \quad &&\varphi_{2,\bD,0}^{(2)} 
&\,=\, \Az_{-1} \ket{\ba_0^{(2)}} \, , \\
-[2,3,2]\colon \quad &&\varphi_{3,\bD,1}^{(2)} 
&\,=\, \Az_{-2} \ket{\ba_1^{(2)}} \, , \\
&&\tilde{\varphi}_{3,\bD,1}^{(2)} 
&\,=\, 2 \sqrt{2} \ \Az_{-2} \ket{\ba_1^{(2)}} 
\,-\, 3 \, \Bz_{-2} \ket{\ba_1^{(2)}}
\,-\, 2 \, \Az_{-1} \Az_{-1} \ket{\ba_1^{(2)}} \, , \\
-[2,3,3]\colon \quad &&\varphi_{3,\bE,0}^{(2)} 
&\,=\, \Az_{-2} \Az_{-1} \ket{\ba_0^{(2)}} \, , \\
&&\varphi_{3,\bD,0}^{(2)} &\,=\, \sqrt{2} \ \Az_{-3} \ket{\ba_0^{(2)}} 
\,-\, 3 \, \Az_{-2} \Az_{-1} \ket{\ba_0^{(2)}} 
\,+\, \sqrt{2} \ \Az_{-1} \Az_{-1} \Az_{-1} \ket{\ba_0^{(2)}} \, , \\
&&\tilde{\varphi}_{3,\bD,0}^{(2)} &\,=\, \Az_{-3} \ket{\ba_0^{(2)}} 
\,+\, 3 \, \Az_{-1} \Bz_{-2} \ket{\ba_0^{(2)}} 
\,-\, 4 \, \Az_{-1} \Az_{-1} \Az_{-1} \ket{\ba_0^{(2)}} \, .
\end{aligned}
\end{align}
From these seven states we can form five commutators that produce
DDF states in the root space of $- [4,5,6]$, namely
\begin{align}
\begin{aligned}
\psi_1 &\,=\, \left[ \varphi_{2,\bD,1}^{(2)}\,,\, \varphi_{3,\bE,0}^{(2)} \right] \, , \\
\psi_2 &\,=\, \left[ \varphi_{2,\bD,1}^{(2)}\,,\, \varphi_{3,\bD,0}^{(2)} \right] \, , \\
\psi_3 &\,=\, \left[ \varphi_{2,\bD,1}^{(2)}\,,\, \tilde{\varphi}_{3,\bD,0}^{(2)} \right] \, , \\
\psi_4 &\,=\, \left[ \varphi_{2,\bD,0}^{(2)}\,,\, \varphi_{3,\bD,1}^{(2)} \right] \, , \\
\psi_5 &\,=\, \left[ \varphi_{2,\bD,0}^{(2)}\,,\, \tilde{\varphi}_{3,\bD,1}^{(2)} \right] \, .
\end{aligned}
\end{align}
\cite{DDFPackage} then tells us that of these five states only four are 
linearly independent. In particular we find
\beq
\psi_1 \,+\, \frac{1}{3} \psi_2 \,+\, \frac{1}{2\sqrt{2}} \psi_3 
\,+\, \psi_4 \,-\, \frac{1}{2\sqrt{2}} \psi_5 \,=\, 0 \, .
\eeq
The root space of $-[4,5,4]$ however, has dimension 5. So we are missing
one state. Since there are no more commutators of level 2 states available that
would give level 4 states in the root space of $-[4,5,4]$ we conclude that
$[\mF^{(2)}, \mF^{(2)}]$ is a proper subset of $\mF^{(4)}$.

\section{Matching characters}
\label{app:delta}

In this appendix we show that the characters of both sides of eqn. 
\eqref{eq:prodL} agree. In particular this fixes all $\bd$-shifts in
\eqref{eq:prodL}. Our starting point is eqn. (4.1) of \cite{KW1}. In our
notation it reads
\begin{gather}
\begin{gathered}
\Ch L(\bL_0) \cdot \Ch L(m\bL_0+2n\bL_1) \\
\,=\, \frac{1}{\varphi(q)} \sum_{k \in K_{m,n}} 
\left( f_k^{(\ell-1,2n)}-f_{2n+1-k}^{(\ell-1,2n)} \right) 
\Ch L\left((m+1+2k)\bL_0 + 2(n-k)\bL_1\right)
\end{gathered}
\end{gather}
with $\ell = m+2n+1$ and
\beq
f_k^{(a,b)} \,=\, 
\sum_{j \in \mathbb{Z}} q^{(a+2)(a+3)j^2 + ((b+1) + 2k(a+2))j + k^2} \, .
\eeq
In \cite{KW1} this equation was derived from the Weyl-Kac character
formula and an identity for the product of $\Theta$-functions 
(see also \cite{KacPeterson}). 
Recalling the coset Virasoro characters \eqref{eq:VirCh}, it is then not hard to see that
\begin{align}
\begin{aligned}
&\quad \frac{1}{\varphi(q)}  \left( f_k^{(\ell-1,2n)}
\,-\,f_{2n+1-k}^{(\ell-1,2n)} \right) 
\,-\, q^{k^2} \chi^{\ell+2,\ell+1}_{2n+1,2n+1-2k}(q) \\
&\,=\, \frac{1}{\varphi(q)}  \sum_{j \in \mathbb{Z}}
\bigg( q^{k^2 + j^2(\ell+2)(\ell+1)+j(1+2n+2k(\ell+1))}  
 \,-\,  q^{ (1-k+2n)^2 + j^2(\ell+2)(\ell+1) + j (1+2n+2(2n+1-k)(\ell+1))}   \\
&\quad \quad \quad \quad \quad \quad
 \,-\,  q^{k^2 + j^2(\ell+2)(\ell+1) + j(1+2n+2k(\ell+1)) }
\,+\,  q^{k^2 + (1+2n+j(\ell+1))(1+2n-2k+j(\ell+2)}  \bigg) \\
&\,=\, 0 \, .
\end{aligned}
\end{align}
Because $\chi^{p,p'}_{r,s}(q) = \chi^{p,p'}_{p'-r,p-s}(q)$ we also find 
\beq
\frac{1}{\varphi(q)} 
\left( f_k^{(\ell-1,2n)}\,-\,f_{2n+1-k}^{(\ell-1,2n)} \right) 
\,=\, q^{k^2} \chi^{\ell+2,\ell+1}_{m+1,m+2+2k}(q)  \, .
\eeq
With $k \in K_{m,n}$ as in \eqref{eq:Kmn} we once again define
$r \equiv r_{m,n,k}$ and $s \equiv s_{m,n,k}$ as in \eqref{eq:rs} and 
subsequently obtain
\begin{gather}
\begin{gathered}
\Ch L(\bL_0) \cdot \Ch L(m\bL_0+2n\bL_1) \\[2mm]
\,=\,  \sum_{k \in K_{m,n}} q^{k^2} \chi_{r,s}^{\ell+2,\ell+1}(q) \ 
\Ch L\left((m+1+2k)\bL_0 + 2(n-k)\bL_1\right) \, .
\end{gathered}
\end{gather}
Using $\Ch L(\bL + l \bd) = q^{-l} \,  \Ch L(\bL)$ and multiplying both sides
with $q^{-2-l}$ we arrive at the $\bd$-shift proposed in \eqref{eq:prodL}.

\end{document}